\documentclass[11pt]{article}

\usepackage[english]{babel}
\usepackage[T1]{fontenc}
\usepackage[utf8]{inputenc}

\usepackage[top=2cm, bottom=2cm, outer=3cm, inner=2cm,
            heightrounded, marginparwidth=2.5cm, marginparsep=1cm]{geometry}

\usepackage{amsmath, amsthm, amssymb}
\usepackage{enumerate}
\usepackage{tikz}
\usetikzlibrary{decorations.pathmorphing}

\usepackage{caption}
\usepackage{float}
\usepackage{framed}
\usepackage{hyperref}
\usepackage{xspace}

\newtheorem{theorem}{Theorem}[section]
\newtheorem{proposition}[theorem]{Proposition}

\newtheorem{conjecture}[theorem]{Conjecture}

\newcommand{\secref}[1]{Section~\ref{#1}\xspace}
\newcommand{\thmref}[1]{Theorem~\ref{#1}\xspace}

\newcommand{\propref}[1]{Proposition~\ref{#1}\xspace}

\newcommand{\conjref}[1]{Conjecture~\ref{#1}\xspace}

\newcommand{\figref}[1]{Figure~\ref{#1}\xspace}

\newcommand{\rt}{\right}
\newcommand{\lt}{\left}

\newcommand{\Prob}{\mathbb{P}}
\newcommand{\E}{\mathbb{E}}
\newcommand{\eps}{\varepsilon}
\newcommand{\mc}[1]{\mathcal #1}
\newcommand{\qd}{qd}
\newcommand{\lvs}{\mathcal L}
\newcommand{\ovl}{\overline}
\newcommand{\tr}{\textup{tr}}

\newcommand*\circled[1]{\tikz[baseline=(char.base)]{
    \node[shape=circle,draw,inner sep=2pt] (char) {#1};}}

\title{On the maximum quartet distance between phylogenetic trees}
\author{
Noga Alon
\thanks{Sackler School of Mathematics and Blavatnik School of
Computer Science, Tel Aviv University, Tel Aviv 69978, Israel and
School
of Mathematics, Institute for Advanced Study, Princeton, NJ 08540.
Email: \texttt{nogaa@tau.ac.il}.
Research supported in part by a USA-Israeli BSF
grant,
by an ISF grant, by the Israeli I-Core program and by the Oswald
Veblen
Fund.}
\and
Humberto Naves
\thanks{Institute for Mathematics and its Applications,
University of Minnesota, Minneapolis, MN 55455, USA.
Email: \texttt{hnaves@ima.umn.edu}.
This research was supported in part by the Institute for
Mathematics and its Applications with funds provided by
the National Science Foundation.} \and
Benny Sudakov
\thanks{Department of Mathematics, ETH, 8092 Zurich.
Email: \texttt{benjamin.sudakov@math.ethz.ch}.
Research supported in part by SNSF grant 200021-149111 and by a
USA-Israeli BSF grant.}}

\date{}

\begin{document}
\maketitle
\setcounter{page}{1}

\begin{abstract}

  A conjecture of Bandelt and Dress states that the maximum quartet distance
  between any two phylogenetic trees on $n$ leaves is at most
  $(\frac 23 +o(1))\binom{n}{4}$.
  Using the machinery of flag algebras we improve the currently known bounds
  regarding this conjecture, in particular we show that the maximum is at most
  $(0.69 +o(1))\binom{n}{4}$.
  We also give further evidence that the conjecture is true by proving that
  the maximum distance between caterpillar trees is at most
  $(\frac 23 +o(1))\binom{n}{4}$.

\end{abstract}

\section{Introduction}

The practice of phylogenetic tree reconstruction to hypothesize various
aspects of evolutionary relationships among different species of organisms
has become a central problem in molecular biology.
For instance, the ``Tree of Life'' project~\cite{TreeOfLife} aims, among
other things, to accurately construct a tree representing the evolutionary
history of the organismal lineages as they change through time.

A phylogeny (the evolutionary history of a set of species) is usually
represented by a tree where the species under study are mapped to the
leaves of the tree and the tree-structure represents the different evolutionary
relationships among them. Here we focus solely on \emph{undirected}
(or \emph{unrooted}) phylogenetic trees. In this setting, the underlying
tree is not directed and each non-leaf node is incident to exactly three
edges. The basic unit of information for phylogenetic classification is the
\emph{quartet}, which is an undirected phylogenetic tree having exactly four
leaves. We denote a quartet over the leaves $\{a,b,c,d\}$ as $[ab|cd]$
whenever there is an edge in the underlying tree separating the pair
$\{a,b\}$ from the pair $\{c,d\}$, as \figref{fig:quartet} shows.
Note that a phylogenetic tree defined over a taxa (species) set of size
$n$ contains the information of exactly $\binom{n}{4}$ quartets.

\begin{figure}[H]
\centering
\begin{tikzpicture}
  [scale=1,auto=left,every node/.style={circle,draw,fill=white,inner sep=1pt}]
  \node (n1) at (-2,  1 cm) {$a$};
  \node (n2) at (-2, -1 cm) {$b$};
  \node (n3) at ( 2,  1 cm) {$c$};
  \node (n4) at ( 2, -1 cm) {$d$};
  \node[inner sep=0pt] (nleft) at (-1, 0 cm) {};
  \node[inner sep=0pt] (nright) at ( 1, 0 cm) {};

  \draw (n1) -- (nleft);
  \draw (n2) -- (nleft);
  \draw (n3) -- (nright);
  \draw (n4) -- (nright);
  \draw (nleft) -- (nright);

\end{tikzpicture}
\caption{A quartet.}
\label{fig:quartet}
\end{figure}
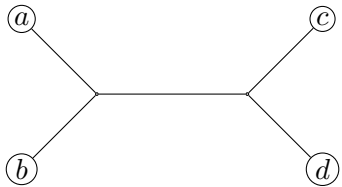

Studying quartets is of prime importance not only because they are the
smallest informational units induced by a phylogeny, but also because they
play a major role in many reconstruction methods. Among them, the
\emph{quartet-based reconstruction} is perhaps the most basic and most
studied approach (see e.g. \cite{BerryGascuel, BerJiaKeaLiWar, JiaKeaLi,
JohWarMorVaw, SnirRao, SnirYuster, StrHae}). The task of the quartet-based
reconstruction is to find a tree over the full set of species that satisfies
most of the given input quartets. In its full generality this problem is very
difficult as Steel~\cite{Steel} has shown that even deciding if there is a tree
that satisfies all the input quartets is NP-complete. To aggravate matters,
even the ideal case in which all quartets agree on a single tree is very
rare. Thus a natural problem arises, namely, finding a tree maximizing the
number of compatible quartets --- \emph{maximum quartet compatibility} (MQC)
\cite{SemSte}. As MQC is obviously NP-hard, several approximation algorithms
have been sugested. However, the best known approximation to the general
problem is still obtained by a naive ``random labelling of the leaves
of a tree'' with expected approximation ratio of $1/3$.

Related to the problem of compatibility is the concept of \emph{quartet
distance} \cite{EstMcMMea}. This notion is used to measure similarity
of two different phylogenetic trees by means of counting how many quartets
are compatible to both of them. More specifically, if $T_1$ and $T_2$ are
two phylogenetic trees on $n$ leaves, let $\qd(T_1,T_2)$ denote the difference
between $\binom{n}{4}$ and the number of quartets compatible to both
$T_1$ and $T_2$. With this definition in mind, a natural question emerges: what
is the maximum quartet distance between two phylogenetic trees on $n$ leaves?
Somewhat surprisingly the answer is strictly smaller than $\binom{n}{4}$.
Bandelt and Dress \cite{BanDre} showed that the maximum is always strictly
smaller than $\frac{14}{15}\binom{n}{4}$ for $n \ge 6$. They also conjectured
that the ratio between the maximum quartet distance and $\binom{n}{4}$
converges to $\frac 23$ as $n$ tends to infinity.

\begin{conjecture}[Bandelt and Dress]

  \label{conj:main}
  The maximum quartet distance between two phylogenetic trees on $n$
  leaves is $\lt(\frac 23 + o(1)\rt) \binom{n}{4}$.

\end{conjecture}

Alon, Snir, and Yuster~\cite{AlonSnirYuster} further improved the bounds
on the maximum quartet distance. Namely, they proved that the maximum is always
strictly larger than $\frac 23 \binom{n}{4}$ but asymptotically smaller
than $\frac{9}{10} \binom{n}{4}$. The lower bound of $\frac 23 \binom{n}{4}$ can
be again obtained by the same ``random labelling of the leaves'' argument,
thus \conjref{conj:main} implies that the average distance between two random
trees is asymptotically the same as the maximum distance.
We also remark that the problem of maximizing the quartet-distance between
trees can be rephrased as how much a compatible set of quartets can be
violated, which is the opposite of MQC.

The main contribution of this paper is the following statement, which 
we obtain using the machinery of flag algebras developed by Razborov
in \cite{Raz_flag}.

\begin{theorem}

  \label{thm:main}
  The maximum quartet distance between two phylogenetic trees on $n$
  leaves is at most $\lt(0.69 + o(1)\rt) \binom{n}{4}$.

\end{theorem}

As further evidence that $\frac 23 \binom{n}{4}$ is the correct answer,
we prove the following statement which establishes \conjref{conj:main}
when restricted to \emph{caterpillar trees}. By \emph{caterpillar}
we mean a phylogenetic tree having at most two vertices which are each
adjacent to two leaves, as in \figref{fig:caterpillar}.

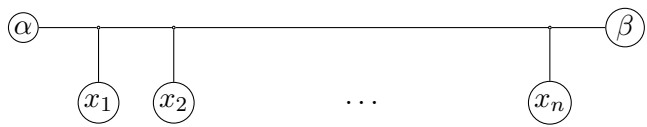
\begin{figure}[H]
\centering
\begin{tikzpicture}
  [scale=1,auto=left,every node/.style={circle,draw,fill=white,inner sep=1pt}]
  \node (nleft) at (-4,  0 cm) {$\alpha$};
  \node (n1) at (-3, -1 cm) {$x_1$};
  \node[inner sep=0pt] (m1) at (-3,  0 cm) {};
  \node (n2) at (-2, -1 cm) {$x_2$};
  \node[inner sep=0pt] (m2) at (-2, 0 cm) {};
  \node (n3) at ( 3, -1 cm) {$x_{n}$};
  \node[inner sep=0pt] (m3) at ( 3,  0 cm) {};
  \node[draw=none,fill=none] (n4) at (0.5, -1 cm) {$\ldots$};
  \node (nright) at ( 4,  0 cm) {$\beta$};

  \draw (nleft) -- (m1);
  \draw (m1) -- (m2);
  \draw (m2) -- (m3);
  \draw (m3) -- (nright);
  \draw (m1) -- (n1);
  \draw (m2) -- (n2);
  \draw (m3) -- (n3);

\end{tikzpicture}
\caption{A caterpillar with $n+2$ leaves.}
\label{fig:caterpillar}
\end{figure}

\begin{theorem}

  \label{thm:caterpillar}
  The maximum quartet distance between two phylogenetic caterpillar trees on
  $n$ leaves is at most $\lt(\frac 23 + o(1)\rt) \binom{n}{4}$.

\end{theorem}

The set of all caterpillar trees is a simple yet very important subclass of
phylogenetic trees. For instance, the proof of NP-hardness of MQC by
Steel~\cite{Steel} heavily relies on this particular subclass. Namely,
deciding if there exists a tree $T$ that satisfies all the quartets
in a given input set is NP-complete even if we further assume that $T$
is caterpillar.

The rest of this paper is organized as follows. In \secref{sec:prelim}, we
formally define all the relevant notions in phylogenetics that were briefly
mentioned in this introduction. In \secref{sec:flag_intro}, we provide an
informal explanation of our main tool, flag algebras. In \secref{sec:main},
we discuss some of the details of the proof of 
\thmref{thm:main}
and provide a link to the program establishing the proof.
In addition,
we give the proof of \thmref{thm:caterpillar} in \secref{sec:caterpillar}.
Lastly, the final section contains some concluding remarks and open problems.

\section{Preliminaries}
\label{sec:prelim}

A \emph{trivalent tree} is a tree in which all internal vertices
(the non-leaves) have exactly three neighbors. Whenever the leaves of
such trees are labeled bijectively by a taxa (species) set $\mc X$
of size $n$, we shall call them \emph{phylogenetic trees}.
Throughout this paper, unless stated otherwise, all trees are assumed
to be phylogenetic trees. For a tree $T= (V,E)$, the set of leaves of
$T$ is denoted by $\lvs(T)$.

The removal of an edge $e$ in a phylogenetic tree splits it
into two subtrees, and thus induces a \emph{split} among the leaves
of the tree. We identify an edge $e$ by the split $(U, \lvs(T)
\setminus U)$ it generates on the set of leaves, and denote this split
by $U_e$. As external edges (the ones adjacent to the leaves) induce
trivial splits, we consider only the ones induced by internal edges.

Let $T$ be a tree and $A\subseteq \lvs(T)$ a subset of the
leaves of $T$. We denote by $T|_{A}$, the topological subtree
of $T$ induced by $A$ were all leaves in $\lvs(T)\setminus A$
and paths leading exclusively to them are removed, and subsequently
internal vertices with degree two are contracted.

For two trees $T$ and $T'$, we say that $T$ \emph{satisfies}
$T'$ (or , equivalently, that $T'$ is \emph{satisfied} by $T$),
if $\lvs(T') \subseteq \lvs(T)$ and $T|_{\lvs(T')} \simeq T'$, that is,
the subgtree of $T$ induced by $\lvs(T')$ is isomorphic to $T'$.
Otherwise, $T'$ is \emph{violated} by $T$. Let $\mc T = \{T_1,
\ldots, T_k\}$ be a set of trees with possibly overlapping leaves,
and denote by $\lvs(\mc T) = \bigcup_i \lvs(T_i)$, the union
of the set of leaves of all trees $T_i \in \mc T$. Then for a
tree $T$ with $\lvs(T) = \lvs(\mc T)$, we denote by $\mc T_s(T)$
the set of trees in $\mc T$ that are satisfied by $T$. We say that
$\mc T$ is \emph{compatible} if there exists a tree $T^*$ over
the set of leaves $\lvs(\mc T)$ that satisfies every tree $T_i\in
\mc T$, i.e. $\mc T_s(T^*)=\mc T$ (see \figref{fig:compatible}).
We denote by $co(\mc T)$ the set of trees that satisfy $\mc T$ (up
to isomorphisms), $co(\mc T) = \{T: \mc T_s(T) = \mc T\}$.

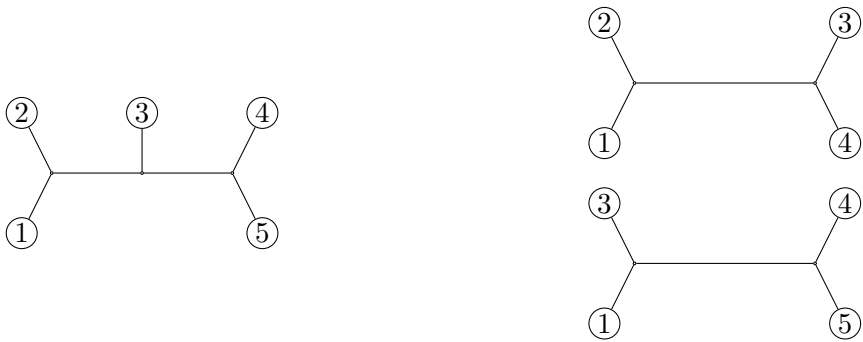
\begin{figure}[H]
\centering
\parbox{3in}{
\centering
\begin{tikzpicture}
  [scale=.8,auto=left,every node/.style={circle,draw,fill=white,inner sep=1pt}]
  \node (n1) at (-2, -1 cm) {$1$};
  \node (n2) at (-2,  1 cm) {$2$};
  \node (n3) at ( 0,  1 cm) {$3$};
  \node (n4) at ( 2,  1 cm) {$4$};
  \node (n5) at ( 2, -1 cm) {$5$};
  \node[inner sep=0pt] (nleft) at ( -1.5,  0 cm) {};
  \node[inner sep=0pt] (nright) at (  1.5,  0 cm) {};
  \node[inner sep=0pt] (nmid) at (  0,  0 cm) {};

  \draw (nleft) -- (nmid);
  \draw (nmid) -- (nright);
  \draw (n1) -- (nleft);
  \draw (n2) -- (nleft);
  \draw (n3) -- (nmid);
  \draw (n4) -- (nright);
  \draw (n5) -- (nright);

\end{tikzpicture}}
\parbox{3in}{
\centering
\begin{tikzpicture}
  [scale=.8,auto=left,every node/.style={circle,draw,fill=white,inner sep=1pt}]
  \node (n1) at (-2,  2 cm) {$1$};
  \node (n2) at (-2,  4 cm) {$2$};
  \node (n3) at ( 2,  4 cm) {$3$};
  \node (n4) at ( 2,  2 cm) {$4$};
  \node[inner sep=0pt] (nleft) at ( -1.5,  3 cm) {};
  \node[inner sep=0pt] (nright) at (  1.5,  3 cm) {};

  \draw (nleft) -- (nright);
  \draw (n1) -- (nleft);
  \draw (n2) -- (nleft);
  \draw (n3) -- (nright);
  \draw (n4) -- (nright);

  \node (m1) at (-2, -1 cm) {$1$};
  \node (m2) at (-2,  1 cm) {$3$};
  \node (m3) at ( 2,  1 cm) {$4$};
  \node (m4) at ( 2, -1 cm) {$5$};
  \node[inner sep=0pt] (mleft) at ( -1.5,  0 cm) {};
  \node[inner sep=0pt] (mright) at (  1.5,  0 cm) {};

  \draw (mleft) -- (mright);
  \draw (m1) -- (mleft);
  \draw (m2) -- (mleft);
  \draw (m3) -- (mright);
  \draw (m4) -- (mright);

\end{tikzpicture}}
\caption{A tree on five leaves and
two quartets compatible with it.}
\label{fig:compatible}
\end{figure}

Further, we say that $T^*$ is \emph{defined} by $\mc T$ if $co(\mc T)$
is the singleton $\{T^*\}$. If there is no such compatible tree $T^*$,
we say that $\mc T$ is \emph{incompatible} (i.e., $co(\mc T) = \emptyset$).

A \emph{quartet} tree (or just a quartet for short), is a phylogenetic
tree over four leaves $\{a, b, c, d\}$. We denote a quartet over
$\{a, b, c, d\}$ as $[ab|cd]$ if there exists an edge $e$ whose
corresponding split $U_e$ satisfies $a,b\in U$ and $c,d\not\in U$.
Quartets are the most elementary informational unit in a phylogenetic
tree, as a pair corresponds to a path in a tree and a triplet to a
vertex (the unique vertex in the intersection of all the pairwise
paths connecting the three leaves). Every phylogenetic tree $T$ with
$n$ leaves defines $\binom{n}{4}$ quartets, one for each set of four
leaves. Let $\mc Q(T)$ denote this full quartet set of $T$.
It is well-known that $\mc Q(T)$ uniquely defines $T$. In fact
Colonius and Schulze~\cite{ColSch} showed that the following proposition
holds.

\begin{proposition}[Colonius and Schulze]

  \label{prop:colsch}
  Let $\mc Q$ be a full quartet set over $n$ species. If every subset of
  three quartets (a quartet triplet) is compatible, then $\mc Q$ is compatible
  and there exists a unique tree defined by $\mc Q$. In fact, if for every five
  taxa $\{a,b,c,d,e\}$ the following holds:
  \begin{equation}
    \label{eqn:consistency}
    \{[ae|bc], [ae|cd]\} \cap \mc Q \ne \emptyset \Rightarrow
    ([ab|cd]\in \mc Q \Rightarrow [be|cd]\in \mc Q),
  \end{equation}
  then $\mc Q$ is consistent.

\end{proposition}

Lastly, we would like to briefly sketch the ``random labelling of the
leaves'' argument. Let $T$ be any tree with $n$ leaves labeled by a taxa set
$\mc X$. Consider a random bijection $\pi$ between $\mc X$ and the leaves
of $T$. The corresponding labeled tree is denoted by $T^\pi$. As each
of the $n!$ possible bijections is equally likely, we notice that a quartet
$[ab|cd]$ with labels from $\mc X$ is satisfied by $T^\pi$ with probability
exactly $1/3$. Thus, by linearity of expectation, we have:

\begin{proposition}

  Let $\mc Q$ be an arbitrary set of quartets over a taxa set $\mc X$ of
  size $n$, and let $T^\pi$ be a random bijection between the leaves of a
  tree $T$ and $\mc X$. The expected number of elements in $\mc Q$ satisfied
  by $T$ is $|\mc Q|/3$.

\end{proposition}

As a consequence, we have the next statement.

\begin{proposition}

  Let $T_1$ and $T_2$ be two random phylogenetic trees over the same taxa set
  $\mc X$ of size $n$, sampled independently and uniformly at random. The
  expected value of the quartet distance $\qd(T_1,T_2)$ is $\frac 23
  \binom{n}{4}$.

\end{proposition}

\section{Flag algebra calculus}
\label{sec:flag_intro}

In this section we provide a brief introduction to the technique of
flag algebras.  First introduced by Razborov in \cite{Raz_flag}, it
has been applied with great success to a wide variety of problems in
extremal combinatorics (see, for example, \cite{Baber, Balogh, DasHuangMaNavSud,
Glebov, Hatami, OlegVaughan, Ravry, Raz_hyper, Raz_tri} and many others).

We begin with a brief explanation on how to map the problem of finding
the maximum quartet distance into a problem in extremal combinatorics.
We then proceed with a general overview of the flag algebra calculus in
the second subsection, by introducing some key definitions and providing
some intuition behind the machinery. The third subsection will show how
we express extremal problems in the language of flag algebras.
It is neither our goal to be rigorous nor thorough, but rather to
emphasize that the combinatorial arguments behind the flag algebra
calculus are as old as extremal combinatorics itself. Indeed, the
main tools available to us are double-counting and the
Cauchy-Schwarz inequality.

The flag algebra calculus is powerful because it provides a
formalism through which the combinatorial problem can be reduced to a
semi-definite programming (SDP) problem. This in turn enables the use of
computers to find solutions, with rigorous proofs, to problems in extremal
combinatorics. For a more complete survey of the technique, we
refer the reader to the excellent expositions in \cite{Keevash} and
\cite{Ravry}, while for a technical specification of flag algebras,
we suggest the original paper of Razborov \cite{Raz_flag}.

\subsection{The model}
\label{sec:model}

In this section, the main object of interest is the \emph{tree-pair}.
From two phylogenetic trees $T_1$ and $T_2$ labelled by the same taxa set,
we would like to create a simple object that ``represents'' the pair
$(T_1,T_2)$ in such a way that we can still compute the quartet distance
$\qd(T_1,T_2)$ from it. Note that the actual set of labels (the taxa set)
is irrelevant in the computation of this distance, so this object shall have no
labels at all. A natural and amenable definition comes to mind. A
\emph{tree-pair} $D$ is a pair of trivalent trees $D=(\ovl{T}_1,
\ovl{T}_2)$ (i.e., unlabelled phylogenetic trees) having the same set of
leaves but having no other vertex in common. In that case we write
$\lvs(D) := \lvs(\ovl{T}_1) = \lvs(\ovl{T}_2)$. From two phylogenetic trees
$T_1$ and $T_2$ over the same taxa set, one can construct a tree-pair
$D$ in the following way. We first identify leaves from $T_1$ and $T_2$ having
the same label and we subsequently remove labels from $T_1$ and $T_2$
altogether to obtain $\ovl{T}_1$ and $\ovl{T}_2$, respectively. We often
represent a tree-pair $D=(\ovl{T}_1, \ovl{T}_2)$ by the graph $\ovl{T}_1
\cup \ovl{T}_2$ which is the \emph{union} of $\ovl{T}_1$ and $\ovl{T}_2$,
that is, $V(\ovl{T}_1\cup \ovl{T}_2) = V(\ovl{T}_1)\cup V(\ovl{T}_2)$
and $E(\ovl{T}_1\cup \ovl{T}_2) = E(\ovl{T}_1) \cup E(\ovl{T}_2)$,
with $\ovl{T}_1$ positioned ``on top'' of $\ovl{T}_2$.
(see \figref{fig:treepair}).

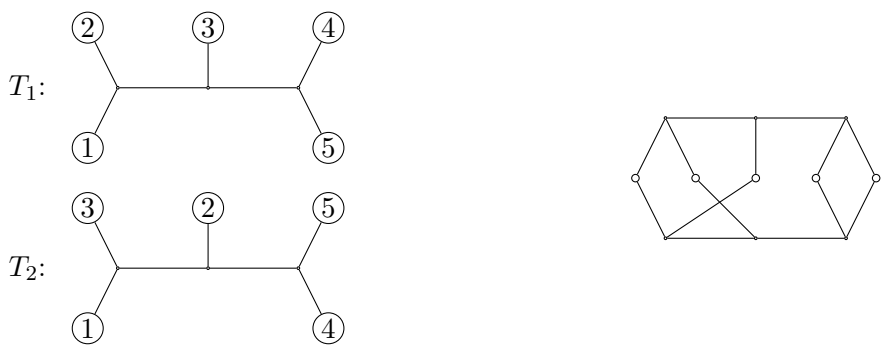
\begin{figure}[H]
\centering
\parbox{3in}{
\centering
\begin{tikzpicture}
  [scale=.8,auto=left,every node/.style={circle,draw,fill=white,inner sep=1pt}]
  \node (n1) at (-2, -1 cm) {$1$};
  \node (n2) at (-2,  1 cm) {$2$};
  \node (n3) at ( 0,  1 cm) {$3$};
  \node (n4) at ( 2,  1 cm) {$4$};
  \node (n5) at ( 2, -1 cm) {$5$};
  \node[inner sep=0pt] (nleft) at ( -1.5,  0 cm) {};
  \node[inner sep=0pt] (nright) at (  1.5,  0 cm) {};
  \node[inner sep=0pt] (nmid) at (  0,  0 cm) {};
  \node[draw=none,fill=none] (nlabel) at (-3, 0 cm) {$T_1$:};

  \draw (nleft) -- (nmid);
  \draw (nmid) -- (nright);
  \draw (n1) -- (nleft);
  \draw (n2) -- (nleft);
  \draw (n3) -- (nmid);
  \draw (n4) -- (nright);
  \draw (n5) -- (nright);

  \node (m1) at (-2, -4 cm) {$1$};
  \node (m2) at (-2, -2 cm) {$3$};
  \node (m3) at ( 0, -2 cm) {$2$};
  \node (m4) at ( 2, -2 cm) {$5$};
  \node (m5) at ( 2, -4 cm) {$4$};
  \node[inner sep=0pt] (mleft) at ( -1.5, -3 cm) {};
  \node[inner sep=0pt] (mright) at (  1.5, -3 cm) {};
  \node[inner sep=0pt] (mmid) at (  0, -3 cm) {};
  \node[draw=none,fill=none] (mlabel) at (-3, -3 cm) {$T_2$:};

  \draw (mleft) -- (mmid);
  \draw (mmid) -- (mright);
  \draw (m1) -- (mleft);
  \draw (m2) -- (mleft);
  \draw (m3) -- (mmid);
  \draw (m4) -- (mright);
  \draw (m5) -- (mright);

\end{tikzpicture}}
\parbox{3in}{
\centering
\begin{tikzpicture}
  [scale=.8,auto=left,every node/.style={circle,draw,fill=white,inner sep=1pt}]
  \node (n1) at (-2,  1 cm) {};
  \node (n2) at (-1,  1 cm) {};
  \node (n3) at ( 0,  1 cm) {};
  \node (n4) at ( 1,  1 cm) {};
  \node (n5) at ( 2,  1 cm) {};
  \node[inner sep=0pt] (nn1) at ( -1.5,  2 cm) {};
  \node[inner sep=0pt] (nn2) at (  0,  2 cm) {};
  \node[inner sep=0pt] (nn3) at (  1.5,  2 cm) {};

  \draw (nn1) -- (nn2);
  \draw (nn2) -- (nn3);
  \draw (n1) -- (nn1);
  \draw (n2) -- (nn1);
  \draw (n3) -- (nn2);
  \draw (n4) -- (nn3);
  \draw (n5) -- (nn3);

  \node[inner sep=0pt] (mm1) at ( -1.5, 0 cm) {};
  \node[inner sep=0pt] (mm2) at (  0, 0 cm) {};
  \node[inner sep=0pt] (mm3) at (  1.5, 0 cm) {};

  \draw (mm1) -- (mm2);
  \draw (mm2) -- (mm3);
  \draw (n1) -- (mm1);
  \draw (n3) -- (mm1);
  \draw (n2) -- (mm2);
  \draw (n5) -- (mm3);
  \draw (n4) -- (mm3);

\end{tikzpicture}}
\caption{Two trees over the same taxa set and the
tree-pair formed by their union.}
\label{fig:treepair}
\end{figure}

Two tree-pairs $D=(\ovl{T}_1, \ovl{T}_2)$ and $D'=(\ovl{T}_1', \ovl{T}_2')$
are \emph{isomorphic} if there exists an isomorphism between the graphs
$\ovl{T}_1\cup \ovl{T}_2$ and $\ovl{T}_1'\cup \ovl{T}_2'$ that maps vertices
of $\ovl{T}_i$ to vertices of $\ovl{T}_i'$ for $i=1,2$, i.e., it respects
the 
component trees. To indicate that two tree-pairs
are isomorphic we write $D \simeq D'$. For a tree-pair $D=(T_1',T_2')$ and a
subset $A\subseteq \lvs(D)$, we write $D|_{A}$ to denote the \emph{sub-tree-pair
induced by $A$}, that is, $D|_{A} := (T_1'|_{A}, T_2'|_{A})$.
Recall that to obtain $T_1'|_{A}$ and $T_2'|_{A}$ we keep the
smallest subtrees of $T_1'$ and $T_2'$ (respectivelly) containing
the leaves in $A$. We subsequently delete one by one all
vertices of degree two by replacing the corresponding path of
length $2$ by an edge until all the vertices not in $A$ have degree
$3$. For two tree-pairs $D$ and $D'$,
if there exists a set $A\subseteq \lvs(D)$ such that $D|_{A} \simeq D'$,
then we say that $D$ contains an \emph{induced copy} of $D'$.

There are only $2$ non-isomorphic tree-pairs having exactly four leaves, namely
$id_4$ and $cr_4$ (see \figref{fig:treepair4}).

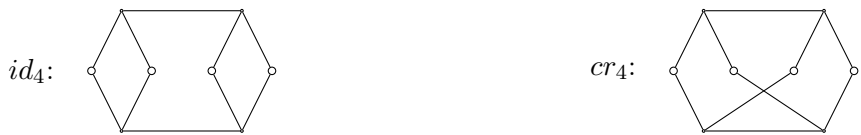
\begin{figure}[H]
\centering
\parbox{3in}{
\centering
\begin{tikzpicture}
  [scale=.8,auto=left,every node/.style={circle,draw,fill=white,inner sep=1pt}]
  \node (n1) at (-1.5,  1 cm) {};
  \node (n2) at (-0.5,  1 cm) {};
  \node (n3) at ( 0.5,  1 cm) {};
  \node (n4) at ( 1.5,  1 cm) {};
  \node[inner sep=0pt] (nn1) at ( -1,  2 cm) {};
  \node[inner sep=0pt] (nn2) at (  1,  2 cm) {};
  \node[draw=none,fill=none] (nlabel) at (-2.5, 1 cm) {$id_4$:};

  \draw (nn1) -- (nn2);
  \draw (n1) -- (nn1);
  \draw (n2) -- (nn1);
  \draw (n3) -- (nn2);
  \draw (n4) -- (nn2);

  \node[inner sep=0pt] (mm1) at ( -1, 0 cm) {};
  \node[inner sep=0pt] (mm2) at (  1, 0 cm) {};

  \draw (mm1) -- (mm2);
  \draw (n1) -- (mm1);
  \draw (n2) -- (mm1);
  \draw (n3) -- (mm2);
  \draw (n4) -- (mm2);

\end{tikzpicture}}
\parbox{3in}{
\centering
\begin{tikzpicture}
  [scale=.8,auto=left,every node/.style={circle,draw,fill=white,inner sep=1pt}]
  \node (n1) at (-1.5,  1 cm) {};
  \node (n2) at (-0.5,  1 cm) {};
  \node (n3) at ( 0.5,  1 cm) {};
  \node (n4) at ( 1.5,  1 cm) {};
  \node[inner sep=0pt] (nn1) at ( -1,  2 cm) {};
  \node[inner sep=0pt] (nn2) at (  1,  2 cm) {};
  \node[draw=none,fill=none] (nlabel) at (-2.5, 1 cm) {$cr_4$:};

  \draw (nn1) -- (nn2);
  \draw (n1) -- (nn1);
  \draw (n2) -- (nn1);
  \draw (n3) -- (nn2);
  \draw (n4) -- (nn2);

  \node[inner sep=0pt] (mm1) at ( -1, 0 cm) {};
  \node[inner sep=0pt] (mm2) at (  1, 0 cm) {};

  \draw (mm1) -- (mm2);
  \draw (n1) -- (mm1);
  \draw (n3) -- (mm1);
  \draw (n2) -- (mm2);
  \draw (n4) -- (mm2);

\end{tikzpicture}}
\caption{The two non-isomorphic tree-pairs on $4$ leaves.}
\label{fig:treepair4}
\end{figure}

A moments thought reveals that the quartet distance $\qd(T_1,T_2)$
can be computed as follows. Let $D=(\ovl{T}_1, \ovl{T}_2)$ be
the tree-pair obtained by joining $T_1$ and $T_2$. Then $\qd(T_1,T_2)$
is simply the number of induced copies of $cr_4$ in $D$. Hence to
prove \thmref{thm:main}, it suffices to show that the following is true.

\begin{theorem}

  \label{thm:main_rephrased}
  Let $D$ be any tree-pair on $n$ leaves. The number of induced
  copies of $cr_4$ in $D$ is at most $(0.69 + o(1))\binom{n}{4}$.

\end{theorem}

\subsection{Definitions and notation in the calculus}
\label{sec:basic_def}

The flag algebra calculus is typically used to find the extremal
density of some fixed structure in a given family of combinatorial
objects. In our case (see \thmref{thm:main_rephrased}), it will be used
to maximize the density of $cr_4$ among all tree-pairs of size $n$,
for $n$ sufficiently large. While the theory of flag algebras is very
general and can be applied to several different types of problems,
we will explain it using only examples related to our particular
setting.

A \emph{type} $\sigma$ is a labeled tree-pair using labels from $[k]$,
where $k = |\lvs(\sigma)|$. That is, each leaf in $\lvs(\sigma)$ is
associated with a label from $[k]$, where $k$ is a nonnegative integer.
The \emph{size} of $\sigma$ is the integer $k$, and is denoted by $|\sigma|$.
\figref{fig:example_types} shows some examples of types.

\vspace{0.1in}
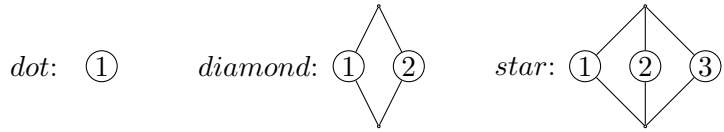
\begin{figure}[H]
\centering

\parbox{1in}{
\centering
\begin{tikzpicture}
  [scale=0.6,auto=left,every node/.style={circle,draw,fill=white,inner sep=1pt}]
  \node (n1) at (0,  0 cm) {$1$};
  \node[draw=none,fill=none] (nlabel) at (-1.5, 0 cm) {$dot$:};

\end{tikzpicture}}
\parbox{1.5in}{
\centering
\begin{tikzpicture}
  [scale=.8,auto=left,every node/.style={circle,draw,fill=white,inner sep=1pt}]
  \node (n1) at (-1.5,  1 cm) {$1$};
  \node (n2) at (-0.5,  1 cm) {$2$};
  \node[inner sep=0pt] (nn1) at ( -1,  2 cm) {};
  \node[draw=none,fill=none] (nlabel) at (-3, 1 cm) {$diamond$:};

  \draw (n1) -- (nn1);
  \draw (n2) -- (nn1);

  \node[inner sep=0pt] (mm1) at ( -1, 0 cm) {};

  \draw (n1) -- (mm1);
  \draw (n2) -- (mm1);

\end{tikzpicture}}
\parbox{1.5in}{
\centering
\begin{tikzpicture}
  [scale=.8,auto=left,every node/.style={circle,draw,fill=white,inner sep=1pt}]
  \node (n1) at (-1,  1 cm) {$1$};
  \node (n2) at (0,  1 cm) {$2$};
  \node (n3) at (1,  1 cm) {$3$};
  \node[inner sep=0pt] (nn1) at ( 0,  2 cm) {};
  \node[draw=none,fill=none] (nlabel) at (-2, 1 cm) {$star$:};

  \draw (n1) -- (nn1);
  \draw (n2) -- (nn1);
  \draw (n3) -- (nn1);

  \node[inner sep=0pt] (mm1) at ( 0, 0 cm) {};

  \draw (n1) -- (mm1);
  \draw (n2) -- (mm1);
  \draw (n3) -- (mm1);

\end{tikzpicture}}
\caption{Examples of types.}
\label{fig:example_types}
\end{figure}

In what follows, an isomorphism between tree-pairs must preserve any labels
that are present. Given a type $\sigma$, a \emph{$\sigma$-flag} is a tree-pair
$F$ on a partially labeled set of leaves, such that the sub-tree-pair induced
by the labeled leaves is isomorphic to $\sigma$. The \emph{underlying tree-pair}
of the flag $F$ is the tree-pair $F$ with all labels removed. The \emph{size}
of a flag is the number of leaves, that is, $|\lvs(F)|$. Note that when $\sigma$
is the \emph{trivial type} of size $0$ (denoted by $\sigma=0$), a $\sigma$-flag
is just a usual unlabeled tree-pair. We shall write $\mc{F}^\sigma_l$ for
the collection of all $\sigma$-flags of size $l$ (up to isomorphism). In
\figref{fig:F4_dot} we list all flags in $\mc{F}^{dot}_4$. Let
$\mc{F}^{\sigma}=\bigcup_{l \ge |\sigma|} \mc{F}^{\sigma}_l$. When the type
$\sigma$ is trivial, we shall omit the superscript from our notation.

\vspace{0.1in}
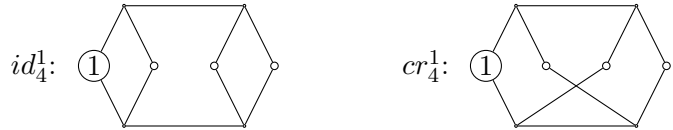
\begin{figure}[H]
\centering
\parbox{2in}{
\centering
\begin{tikzpicture}
  [scale=.8,auto=left,every node/.style={circle,draw,fill=white,inner sep=1pt}]
  \node (n1) at (-1.5,  1 cm) {$1$};
  \node (n2) at (-0.5,  1 cm) {};
  \node (n3) at ( 0.5,  1 cm) {};
  \node (n4) at ( 1.5,  1 cm) {};
  \node[inner sep=0pt] (nn1) at ( -1,  2 cm) {};
  \node[inner sep=0pt] (nn2) at (  1,  2 cm) {};
  \node[draw=none,fill=none] (nlabel) at (-2.5, 1 cm) {$id_4^1$:};

  \draw (nn1) -- (nn2);
  \draw (n1) -- (nn1);
  \draw (n2) -- (nn1);
  \draw (n3) -- (nn2);
  \draw (n4) -- (nn2);

  \node[inner sep=0pt] (mm1) at ( -1, 0 cm) {};
  \node[inner sep=0pt] (mm2) at (  1, 0 cm) {};

  \draw (mm1) -- (mm2);
  \draw (n1) -- (mm1);
  \draw (n2) -- (mm1);
  \draw (n3) -- (mm2);
  \draw (n4) -- (mm2);

\end{tikzpicture}}
\parbox{2in}{
\centering
\begin{tikzpicture}
  [scale=.8,auto=left,every node/.style={circle,draw,fill=white,inner sep=1pt}]
  \node (n1) at (-1.5,  1 cm) {$1$};
  \node (n2) at (-0.5,  1 cm) {};
  \node (n3) at ( 0.5,  1 cm) {};
  \node (n4) at ( 1.5,  1 cm) {};
  \node[inner sep=0pt] (nn1) at ( -1,  2 cm) {};
  \node[inner sep=0pt] (nn2) at (  1,  2 cm) {};
  \node[draw=none,fill=none] (nlabel) at (-2.5, 1 cm) {$cr_4^1$:};

  \draw (nn1) -- (nn2);
  \draw (n1) -- (nn1);
  \draw (n2) -- (nn1);
  \draw (n3) -- (nn2);
  \draw (n4) -- (nn2);

  \node[inner sep=0pt] (mm1) at ( -1, 0 cm) {};
  \node[inner sep=0pt] (mm2) at (  1, 0 cm) {};

  \draw (mm1) -- (mm2);
  \draw (n1) -- (mm1);
  \draw (n3) -- (mm1);
  \draw (n2) -- (mm2);
  \draw (n4) -- (mm2);

\end{tikzpicture}}
\caption{Family $\mc{F}^{dot}_4$.}
\label{fig:F4_dot}
\end{figure}

\medskip

Let us now define two fundamental concepts in our calculus, namely
those of flag densities in larger flags and tree-pairs. Let $\sigma$ be
a type of size $k$, let $m\ge 1$ be an integer and let
\{$F_i\}_{i=1}^m$ be a collection of $\sigma$-flags of sizes $l_i =
|F_i| \ge k$. Given a $\sigma$-flag $F$ of order at least $l = k +
\sum_{i=1}^m (l_i-k)$, let $A \subseteq \lvs(F)$ be the set of labeled
leaves of $F$. Now select disjoint subsets $X_i\subseteq \lvs(F)
\setminus A$ of sizes $|X_i| = l_i-k$, uniformly at random. This is
possible because $F$ has at least $\sum_i (l_i - k)$ unlabeled
leaves. Denote by $E_i$ the event that the $\sigma$-flag induced
by $A \cup X_i$ is isomorphic to $F_i$, for $i\in [m]$. We define
$p_\sigma(F_1,F_2,\ldots,F_m; F):=
\Prob\lt[\cap_{i=1}^m E_i\rt]$ to be the probability that all these
events occur simultaneously.

If $D$ is just a tree-pair of order at least $l$, and not a
$\sigma$-flag, then there is no pre-labeled set of leaves $A$ that
induces the type $\sigma$.  Instead, we uniformly at random select a
partial labeling $L : [k] \rightarrow \lvs(D)$. This random labeling
turns $D$ into a $\sigma'$-flag $F_L$, where the type $\sigma'$ is
the labeled sub-tree-pair induced by the set of vertices $L([k])$.
If $\sigma' = \sigma$, we can then
proceed as above, otherwise we say the events $E_i$ have probability
$0$. Finally, we average over all possible random labellings.
Formally, let $Y$ be the following random variable
\[
  Y := \left\{\begin{array}{ll}
  p_{\sigma}(F_1,F_2,\ldots, F_m; F_L) & \text{if } \sigma'=\sigma \\
  0 & \text{otherwise}
  \end{array}\right..
\]
Define $d_{\sigma}(F_1,\ldots, F_m;D):=\E[Y]$ as the
expected value of the random variable $Y$. The quantities
$p_\sigma(F_1,F_2,\ldots, F_m; F)$ and $d_\sigma(F_1,F_2,\ldots, F_m;D)$
are called \emph{flag densities} of $\{F_i\}_{i\in[m]}$ in $F$ and in $D$,
respectively. Clearly these flag densities are the same whenever $\sigma=0$,
in which case we omit the subscript from both notations.

\medskip

To better illustrate these definitions, we give some examples.
These flags are shown in \figref{fig:example_flags}.

\vspace{0.1in}
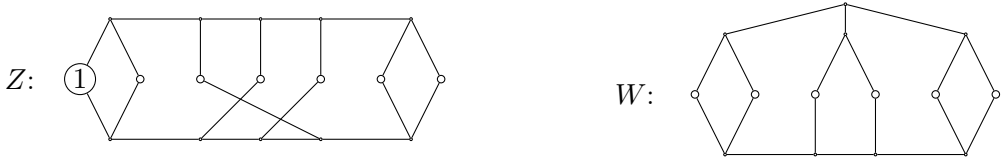
\begin{figure}[H]
\centering

\parbox{3in}{
\centering
\begin{tikzpicture}
  [scale=.8,auto=left,every node/.style={circle,draw,fill=white,inner sep=1pt}]
  \node (n1) at (  -3,  1 cm) {$1$};
  \node (n2) at (  -2,  1 cm) {};
  \node (n3) at (  -1,  1 cm) {};
  \node (n4) at (   0,  1 cm) {};
  \node (n5) at (   1,  1 cm) {};
  \node (n6) at (   2,  1 cm) {};
  \node (n7) at (   3,  1 cm) {};
  \node[inner sep=0pt] (nn1) at ( -2.5,  2 cm) {};
  \node[inner sep=0pt] (nn2) at ( -1,  2 cm) {};
  \node[inner sep=0pt] (nn3) at (  0,  2 cm) {};
  \node[inner sep=0pt] (nn4) at (  1,  2 cm) {};
  \node[inner sep=0pt] (nn5) at (  2.5,  2 cm) {};
  \node[draw=none,fill=none] (nlabel) at (-4, 1 cm) {$Z$:};

  \draw (nn1) -- (nn2);
  \draw (nn2) -- (nn3);
  \draw (nn3) -- (nn4);
  \draw (nn4) -- (nn5);
  \draw (n1) -- (nn1);
  \draw (n2) -- (nn1);
  \draw (n3) -- (nn2);
  \draw (n4) -- (nn3);
  \draw (n5) -- (nn4);
  \draw (n6) -- (nn5);
  \draw (n7) -- (nn5);

  \node[inner sep=0pt] (mm1) at ( -2.5, 0 cm) {};
  \node[inner sep=0pt] (mm2) at ( -1, 0 cm) {};
  \node[inner sep=0pt] (mm3) at (  0, 0 cm) {};
  \node[inner sep=0pt] (mm4) at (  1, 0 cm) {};
  \node[inner sep=0pt] (mm5) at (  2.5, 0 cm) {};

  \draw (mm1) -- (mm2);
  \draw (mm2) -- (mm3);
  \draw (mm3) -- (mm4);
  \draw (mm4) -- (mm5);
  \draw (n1) -- (mm1);
  \draw (n2) -- (mm1);
  \draw (n4) -- (mm2);
  \draw (n5) -- (mm3);
  \draw (n3) -- (mm4);
  \draw (n6) -- (mm5);
  \draw (n7) -- (mm5);

\end{tikzpicture}}
\parbox{3in}{
\centering
\begin{tikzpicture}
  [scale=.8,auto=left,every node/.style={circle,draw,fill=white,inner sep=1pt}]
  \node (n1) at (  -3,  1 cm) {};
  \node (n2) at (  -2,  1 cm) {};
  \node (n3) at (  -1,  1 cm) {};
  \node (n4) at (   0,  1 cm) {};
  \node (n5) at (   1,  1 cm) {};
  \node (n6) at (   2,  1 cm) {};
  \node[inner sep=0pt] (nn1) at ( -2.5,  2 cm) {};
  \node[inner sep=0pt] (nn2) at ( -0.5,  2 cm) {};
  \node[inner sep=0pt] (nn3) at (  1.5,  2 cm) {};
  \node[inner sep=0pt] (nntop) at (  -0.5,  2.5 cm) {};
  \node[draw=none,fill=none] (nlabel) at (-4, 1 cm) {$W$:};

  \draw (nn1) -- (nntop);
  \draw (nn2) -- (nntop);
  \draw (nn3) -- (nntop);
  \draw (n1) -- (nn1);
  \draw (n2) -- (nn1);
  \draw (n3) -- (nn2);
  \draw (n4) -- (nn2);
  \draw (n5) -- (nn3);
  \draw (n6) -- (nn3);

  \node[inner sep=0pt] (mm1) at ( -2.5, 0 cm) {};
  \node[inner sep=0pt] (mm2) at ( -1, 0 cm) {};
  \node[inner sep=0pt] (mm3) at (  0, 0 cm) {};
  \node[inner sep=0pt] (mm4) at (  1.5, 0 cm) {};

  \draw (mm1) -- (mm2);
  \draw (mm2) -- (mm3);
  \draw (mm3) -- (mm4);
  \draw (n1) -- (mm1);
  \draw (n2) -- (mm1);
  \draw (n3) -- (mm2);
  \draw (n4) -- (mm3);
  \draw (n5) -- (mm4);
  \draw (n6) -- (mm4);

\end{tikzpicture}}
\caption{Example flags.}
\label{fig:example_flags}
\end{figure}

We turn to compute the flag densities of $id_4^1$ and $cr_4^1$ in
the flag $Z$. For example, to compute $p_{dot}(id_4^1;
Z)$, note that to induce a copy of $id_4^1$ we must choose exactly
3 other unlabeled leaves which together with the labeled leaf $1$
induce a copy of $id_4^1$. There are $\binom{6}{3}=20$ ways to make the
choice of 3 unlabeled leaves, and out of the $20$ exactly $15$ induce
a copy of $id_4^1$, thus $p_{dot}(id_4^1;Z)=\frac 34$. Similarly
we obtain $p_{dot}(cr_4^1;Z) = \frac 14$.
We can also compute the joint flag densities of multiple flags. For
instance, let us consider $p_{dot}(id_4^1, id_4^1; Z)$. In this case,
we first choose a set $X_1$ of $3$ unlabeled leaves uniformly at random
and we subsequently choose a set $X_2$ of $3$ other unlabeled leaves
also uniformly at random. Since the choice of $X_2$ is uniquely determined
given the choice of $X_1$, there are exactly $\binom{6}{3}$ possible choices
for the pair $(X_1,X_2)$. Out of these $20$ choices, one can
count that exactly $10$ of them will be such that both $X_1$ and $X_2$
will induce a copy of $id_4^1$ when we add the labeled leaf.
The order matters here: when computing $p(F_1,F_2;F)$, the set $X_1$ must
induce a copy of $F_1$ while $X_2$ must induce a copy of 
$F_2$.
Thus $p_{dot}(id_4^1, id_4^1; Z) = \frac12$. Similarly, we have
$p_{dot}(id_4^1, cr_4^1; Z) = p_{dot}(cr_4^1,id_4^1;Z)= \frac 14$ and
$p_{dot}(cr_4^1, cr_4^1; Z) = 0$.

The computation of flag densities $d_{dot}$ for unlabeled tree-pairs
is a little more involved. To see how to compute it, we consider
$W$ depicted in \figref{fig:example_flags} as an example.
There are two non-isomorphic $dot$-flags whose underlying tree-pair
is $W$, namely $W_1^1$ and $W_2^1$ as shown in \figref{fig:flags_of_w}.

\vspace{0.1in}
\begin{figure}[H]
\centering
\parbox{2.5in}{
\centering
\begin{tikzpicture}
  [scale=.8,auto=left,every node/.style={circle,draw,fill=white,inner sep=1pt}]
  \node (n1) at (  -3,  1 cm) {$1$};
  \node (n2) at (  -2,  1 cm) {};
  \node (n3) at (  -1,  1 cm) {};
  \node (n4) at (   0,  1 cm) {};
  \node (n5) at (   1,  1 cm) {};
  \node (n6) at (   2,  1 cm) {};
  \node[inner sep=0pt] (nn1) at ( -2.5,  2 cm) {};
  \node[inner sep=0pt] (nn2) at ( -0.5,  2 cm) {};
  \node[inner sep=0pt] (nn3) at (  1.5,  2 cm) {};
  \node[inner sep=0pt] (nntop) at ( -0.5,  2.5 cm) {};
  \node[draw=none,fill=none] (nlabel) at (-4, 1 cm) {$W_1^1$:};

  \draw (nn1) -- (nntop);
  \draw (nn2) -- (nntop);
  \draw (nn3) -- (nntop);
  \draw (n1) -- (nn1);
  \draw (n2) -- (nn1);
  \draw (n3) -- (nn2);
  \draw (n4) -- (nn2);
  \draw (n5) -- (nn3);
  \draw (n6) -- (nn3);

  \node[inner sep=0pt] (mm1) at ( -2.5, 0 cm) {};
  \node[inner sep=0pt] (mm2) at ( -1, 0 cm) {};
  \node[inner sep=0pt] (mm3) at (  0, 0 cm) {};
  \node[inner sep=0pt] (mm4) at (  1.5, 0 cm) {};

  \draw (mm1) -- (mm2);
  \draw (mm2) -- (mm3);
  \draw (mm3) -- (mm4);
  \draw (n1) -- (mm1);
  \draw (n2) -- (mm1);
  \draw (n3) -- (mm2);
  \draw (n4) -- (mm3);
  \draw (n5) -- (mm4);
  \draw (n6) -- (mm4);

\end{tikzpicture}}
\parbox{2.5in}{
\centering
\begin{tikzpicture}
  [scale=.8,auto=left,every node/.style={circle,draw,fill=white,inner sep=1pt}]
  \node (n1) at (  -3,  1 cm) {};
  \node (n2) at (  -2,  1 cm) {};
  \node (n3) at (  -1,  1 cm) {$1$};
  \node (n4) at (   0,  1 cm) {};
  \node (n5) at (   1,  1 cm) {};
  \node (n6) at (   2,  1 cm) {};
  \node[inner sep=0pt] (nn1) at ( -2.5,  2 cm) {};
  \node[inner sep=0pt] (nn2) at ( -0.5,  2 cm) {};
  \node[inner sep=0pt] (nn3) at (  1.5,  2 cm) {};
  \node[inner sep=0pt] (nntop) at ( -0.5,  2.5 cm) {};
  \node[draw=none,fill=none] (nlabel) at (-4, 1 cm) {$W_2^1$:};

  \draw (nn1) -- (nntop);
  \draw (nn2) -- (nntop);
  \draw (nn3) -- (nntop);
  \draw (n1) -- (nn1);
  \draw (n2) -- (nn1);
  \draw (n3) -- (nn2);
  \draw (n4) -- (nn2);
  \draw (n5) -- (nn3);
  \draw (n6) -- (nn3);

  \node[inner sep=0pt] (mm1) at ( -2.5, 0 cm) {};
  \node[inner sep=0pt] (mm2) at ( -1, 0 cm) {};
  \node[inner sep=0pt] (mm3) at (  0, 0 cm) {};
  \node[inner sep=0pt] (mm4) at (  1.5, 0 cm) {};

  \draw (mm1) -- (mm2);
  \draw (mm2) -- (mm3);
  \draw (mm3) -- (mm4);
  \draw (n1) -- (mm1);
  \draw (n2) -- (mm1);
  \draw (n3) -- (mm2);
  \draw (n4) -- (mm3);
  \draw (n5) -- (mm4);
  \draw (n6) -- (mm4);

\end{tikzpicture}}
\caption{$dot$-flags for $W$.}
\label{fig:flags_of_w}
\end{figure}
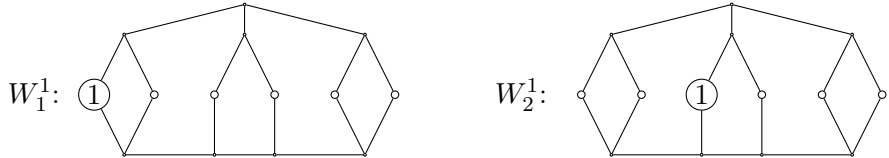

If we randomly label a leaf from $W$, then with probability $\frac 23$
it will become $W_1^1$ and with probability $\frac 13$ it will become $W_2^1$.
Moreover, since $p(id_4^1; W_1^1) = \frac45$ and $p_{dot}(id_4^1; W_2^1)=
\frac 35$, we have $d_{dot}(id_4^1; W) = \frac23 p_{dot}(id_4^1; W_1^1)+\frac13
p_{dot}(id_4^1; W_2^1)=\frac{11}{15}$. Similarly we have
$p_{dot}(cr_4^{1}; W_1^1)=\frac 15$ and $p_{dot}(cr_4^{1}; W_2^1)=\frac 25$,
hence $d_{dot}(cr_4^{1}; W)=\frac{4}{15}$.

The reader might notice that there is an alternative way to compute,
say, $d_{dot}(cr_4^{1}; W)$: we simply compute the product of
$d_{dot}(cr_4^{1}; cr_4) \cdot p(cr_4; W)=1\cdot\frac{4}{15}=\frac{4}{15}$.
In general, suppose as before that we
have a type $\sigma$ of size $k$, a $\sigma$-flag $F$ of size $l \ge k$,
and an unlabeled tree-pair $D$. To compute $d_{\sigma}(F ; D)$, we averaged
over all random partial labelings of $D$ the probability of finding a flag
isomorphic to $F$. A simple double-counting argument shows that we can do the
``averaging'' before the random labeling, which is the idea behind
Razborov's \emph{averaging operator}, as defined in Section 2.2 of
\cite{Raz_flag}. Let $F|_0$ denote the unlabeled underlying model of
$F$. We can compute $d_{\sigma}(F;D)$ by first computing
$d(F|_0;D)$, the probability that $l$ randomly chosen vertices in
$D$ form an induced copy of $F|_0$ as a sub-model. Given this copy
of $F|_0$, we then randomly label $k$ of the $l$ vertices, and
compute the probability that these $k$ vertices are label-isomorphic
to $\sigma$.  This amounts to multiplying $d(F|_0; D)$ by a
\emph{normalizing factor} $q_{\sigma}(F)$, that is, $d_{\sigma}(F;D)
= q_{\sigma}(F) d(F|_0;D)= q_{\sigma}(F) p(F|_0;D)$.
We can interpret the normalizing factor as $q_{\sigma}(F) =
d_{\sigma}(F ; F|_0)$.

\medskip

There are more relations involving $d_\sigma$ and $p_\sigma$ than
the one mentioned previously. We will now state, without proof, a
basic fact about flag densities that can be proved easily by double
counting.

\begin{proposition}[Chain rule]

  \label{prop:chain_rule}
  If $\sigma$ is a type of size $k$, $m\ge 1$ is an integer, and
  $\{F_i\}_{i=1}^m$ is a family of $\sigma$-flags of sizes $|F_i|=l_i$,
  and $l \ge k + \sum_{i=1}^m (l_i - k)$ is an integer parameter, then
  \begin{enumerate}
    \item For any $\sigma$-flag $F$ of order at least $l$, we have
    \[
      p_{\sigma}(F_1,\ldots,F_m;F) = \sum_{F' \in \mc{F}^{\sigma}_{l}}
      p_{\sigma}(F_1,\ldots, F_m;F') p_{\sigma}(F';F).
    \]
    \item For any tree-pair $D$ of size at least $l$, we have
    \[
      d_{\sigma}(F_1,\ldots,F_m;D) = \sum_{H \in \mc{F}_{l}}
      d_{\sigma}(F_1,\ldots,F_m;H) d(H;D)=
      \sum_{F \in \mc{F}^{\sigma}_{l}}
        p_{\sigma}(F_1,\ldots, F_m;F) d_{\sigma}(F;D).
    \]
  \end{enumerate}

\end{proposition}

To illustrate the chain rule for $m=1$ and $\sigma=0$, we consider the
``expansion'' of $id_4$ in $\mc F_5$ (see \figref{fig:F_5}).

\vspace{0.1in}
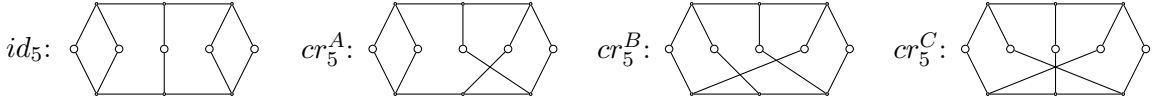
\begin{figure}[H]
\centering
\parbox{1.5in}{
\centering
\begin{tikzpicture}
  [scale=.6,auto=left,every node/.style={circle,draw,fill=white,inner sep=1pt}]
  \node (n1) at (  -2,  1 cm) {};
  \node (n2) at (  -1,  1 cm) {};
  \node (n3) at (   0,  1 cm) {};
  \node (n4) at (   1,  1 cm) {};
  \node (n5) at (   2,  1 cm) {};
  \node[inner sep=0pt] (nn1) at ( -1.5,  2 cm) {};
  \node[inner sep=0pt] (nn2) at (  0,  2 cm) {};
  \node[inner sep=0pt] (nn3) at (  1.5,  2 cm) {};
  \node[draw=none,fill=none] (nlabel) at (-3, 1 cm) {$id_5$:};

  \draw (nn1) -- (nn2);
  \draw (nn2) -- (nn3);
  \draw (n1) -- (nn1);
  \draw (n2) -- (nn1);
  \draw (n3) -- (nn2);
  \draw (n4) -- (nn3);
  \draw (n5) -- (nn3);

  \node[inner sep=0pt] (mm1) at ( -1.5, 0 cm) {};
  \node[inner sep=0pt] (mm2) at (  0, 0 cm) {};
  \node[inner sep=0pt] (mm3) at (  1.5, 0 cm) {};

  \draw (mm1) -- (mm2);
  \draw (mm2) -- (mm3);
  \draw (n1) -- (mm1);
  \draw (n2) -- (mm1);
  \draw (n3) -- (mm2);
  \draw (n4) -- (mm3);
  \draw (n5) -- (mm3);

\end{tikzpicture}}
\parbox{1.5in}{
\centering
\begin{tikzpicture}
  [scale=.6,auto=left,every node/.style={circle,draw,fill=white,inner sep=1pt}]
  \node (n1) at (  -2,  1 cm) {};
  \node (n2) at (  -1,  1 cm) {};
  \node (n3) at (   0,  1 cm) {};
  \node (n4) at (   1,  1 cm) {};
  \node (n5) at (   2,  1 cm) {};
  \node[inner sep=0pt] (nn1) at ( -1.5,  2 cm) {};
  \node[inner sep=0pt] (nn2) at (  0,  2 cm) {};
  \node[inner sep=0pt] (nn3) at (  1.5,  2 cm) {};
  \node[draw=none,fill=none] (nlabel) at (-3, 1 cm) {$cr_5^A$:};

  \draw (nn1) -- (nn2);
  \draw (nn2) -- (nn3);
  \draw (n1) -- (nn1);
  \draw (n2) -- (nn1);
  \draw (n3) -- (nn2);
  \draw (n4) -- (nn3);
  \draw (n5) -- (nn3);

  \node[inner sep=0pt] (mm1) at ( -1.5, 0 cm) {};
  \node[inner sep=0pt] (mm2) at (  0, 0 cm) {};
  \node[inner sep=0pt] (mm3) at (  1.5, 0 cm) {};

  \draw (mm1) -- (mm2);
  \draw (mm2) -- (mm3);
  \draw (n1) -- (mm1);
  \draw (n2) -- (mm1);
  \draw (n4) -- (mm2);
  \draw (n3) -- (mm3);
  \draw (n5) -- (mm3);

\end{tikzpicture}}
\parbox{1.5in}{
\centering
\begin{tikzpicture}
  [scale=.6,auto=left,every node/.style={circle,draw,fill=white,inner sep=1pt}]
  \node (n1) at (  -2,  1 cm) {};
  \node (n2) at (  -1,  1 cm) {};
  \node (n3) at (   0,  1 cm) {};
  \node (n4) at (   1,  1 cm) {};
  \node (n5) at (   2,  1 cm) {};
  \node[inner sep=0pt] (nn1) at ( -1.5,  2 cm) {};
  \node[inner sep=0pt] (nn2) at (  0,  2 cm) {};
  \node[inner sep=0pt] (nn3) at (  1.5,  2 cm) {};
  \node[draw=none,fill=none] (nlabel) at (-3, 1 cm) {$cr_5^B$:};

  \draw (nn1) -- (nn2);
  \draw (nn2) -- (nn3);
  \draw (n1) -- (nn1);
  \draw (n2) -- (nn1);
  \draw (n3) -- (nn2);
  \draw (n4) -- (nn3);
  \draw (n5) -- (nn3);

  \node[inner sep=0pt] (mm1) at ( -1.5, 0 cm) {};
  \node[inner sep=0pt] (mm2) at (  0, 0 cm) {};
  \node[inner sep=0pt] (mm3) at (  1.5, 0 cm) {};

  \draw (mm1) -- (mm2);
  \draw (mm2) -- (mm3);
  \draw (n1) -- (mm1);
  \draw (n4) -- (mm1);
  \draw (n2) -- (mm2);
  \draw (n3) -- (mm3);
  \draw (n5) -- (mm3);

\end{tikzpicture}}
\parbox{1.5in}{
\centering
\begin{tikzpicture}
  [scale=.6,auto=left,every node/.style={circle,draw,fill=white,inner sep=1pt}]
  \node (n1) at (  -2,  1 cm) {};
  \node (n2) at (  -1,  1 cm) {};
  \node (n3) at (   0,  1 cm) {};
  \node (n4) at (   1,  1 cm) {};
  \node (n5) at (   2,  1 cm) {};
  \node[inner sep=0pt] (nn1) at ( -1.5,  2 cm) {};
  \node[inner sep=0pt] (nn2) at (  0,  2 cm) {};
  \node[inner sep=0pt] (nn3) at (  1.5,  2 cm) {};
  \node[draw=none,fill=none] (nlabel) at (-3, 1 cm) {$cr_5^C$:};

  \draw (nn1) -- (nn2);
  \draw (nn2) -- (nn3);
  \draw (n1) -- (nn1);
  \draw (n2) -- (nn1);
  \draw (n3) -- (nn2);
  \draw (n4) -- (nn3);
  \draw (n5) -- (nn3);

  \node[inner sep=0pt] (mm1) at ( -1.5, 0 cm) {};
  \node[inner sep=0pt] (mm2) at (  0, 0 cm) {};
  \node[inner sep=0pt] (mm3) at (  1.5, 0 cm) {};

  \draw (mm1) -- (mm2);
  \draw (mm2) -- (mm3);
  \draw (n1) -- (mm1);
  \draw (n4) -- (mm1);
  \draw (n3) -- (mm2);
  \draw (n2) -- (mm3);
  \draw (n5) -- (mm3);

\end{tikzpicture}}
\caption{Family $\mc F_5$.}
\label{fig:F_5}
\end{figure}
The chain rule gives
\begin{align*}
    p(id_4; F) &= p(id_4; id_5) p(id_5; F) + p(id_4; cr_5^A)
    p(cr_5^A; F)+ \\
    &+p(id_4; cr_5^B) p(cr_5^B; F) +
    p(id_5; cr_5^C) p(cr_5^C; F)\\
    &= p(id_5; F) + \frac 35 p(cr_5^A; F) +
    \frac 15 p(cr_5^B; F).
\end{align*}

Similarly, we can expand $p(cr_4; F) = \frac 25
p(cr_5^A;F) + \frac 45 p(cr_5^B; F) + p(cr_5^C; F)$.

For the ease of notation, we can express these two identities using
the syntax of flag algebras:
\begin{align*}
  id_4 &= id_5 + \frac35 cr_5^A + \frac 15 cr_5^B \\
  cr_4 &= \frac25 cr_5^A + \frac45 cr_5^B + cr_5^C.
\end{align*}
In this syntax, the equation $\sum_{i\in I} \alpha_i F_i = 0$
means that for all sufficiently large $\sigma$-flags $F$, we have
$\sum_{i\in I} \alpha_i p_{\sigma}(F_i; F) = 0$, where $\alpha_i \in
\mathbb{R}$ for all $i\in I$. We use $\mc{A}^{\sigma}$ to
denote the set of linear combinations of flags of type $\sigma$.
It is convenient to define a \emph{product} of flags in the following way:
\[
  F_1\cdot F_2 := \sum_{F\in \mc{F}^{\sigma}_l}
  p_{\sigma}(F_1,F_2; F) F, \qquad F_1 \in \mc{F}^{\sigma},
    F_2 \in \mc{F}^{\sigma}, l \ge |F_1| + |F_2| - |\sigma|.
\]
(Note that because of the chain rule, it does not matter which $l$ we
choose.)

To further simplify the notation, we can extend the definitions of
$p_{\sigma}$ and $d_{\sigma}$ to $\mc{A}^{\sigma}$ by making them
linear in each coordinate.
The product notation simplifies these extended definitions, because
$p_{\sigma}(f_1\cdot f_2; f) = p_{\sigma}(f_1,f_2;f)$ and
$d_{\sigma}(f_1\cdot f_2; g) = d_{\sigma}(f_1, f_2; g)$, for any $f_1,f_2,f\in
\mc{A}^{\sigma}$ and for any $g\in\mc{A}^0$.

\medskip

The last piece of notation we introduce is that of the averaging
operator.  Recall that for any $\sigma$-flag $F$, we had the
normalizing factors $q_{\sigma}(F)$ such that $d_{\sigma}(F; G) =
q_{\sigma}(F) p(F|_0 ; G)$. In the syntax of flag algebra, this
averaging operation is denoted by
$[[F]]_{\sigma}:=q_{\sigma}(F)\cdot F|_0$. We extend
this linearly to all elements of $\mc{A}^{\sigma}$.
For example
\[
  [[id_4^1]]_{dot} = id_4, \quad [[cr_4^1]]_{dot} = cr_4,
  \quad [[id_4^1 + cr_4^1]]_{dot} = id_4 + cr_4, \text{ and } \quad
  [[W_1^1]]_{dot} = \frac 23 W.
\]
This notation is useful, because $d_{\sigma}(f; g)=p([[f]]_{\sigma}; g)$
for any $f\in\mc{A}^{\sigma}$ and for any $g\in\mc{A}^0$, and
hence we have a unified notation for both types of flag densities.

\subsection{Extremal problems in the flag algebra calculus}
\label{sec:extremal_prob}

Recall that our optimization problem is to maximize the density of
$cr_4$ amongst all possible tree-pairs. We will show how
flag algebras can be applied to this problem to reduce it to a
semi-definite programming (SDP) problem, which can then be
solved numerically.

We may use the chain rule to obtain $d(cr_4;D) =
\sum_{H \in \mc{F}_t} d(cr_4;H) d(H;D)$ for $t\ge 4$.
Since $\sum_{H\in \mc{F}_t} d(H;D)=1$, we have
\[
  d(cr_4;D) \le \max_{H\in \mc{F}_t} d(cr_4;H),
\]
which is a bound that clearly does not depend on $D$. For instance,
when we choose $t=6$ we already obtain $d(cr_4;D) \le \frac{14}{15}$.

Inequalities obtained 
this way are often very weak, since we only use very local
considerations about the sub-tree-pairs $H \in \mc{F}_t$, and we do
not take into account how the tree-pairs fit together in the larger
tree-pair $D$; that is, how they intersect.

One might hope to find inequalities of the form $\sum_{H \in\mc{F}_t}
\alpha_H d(H;D) \ge 0$, such that when we combine them with the initial
identity, we get
\begin{align*}
  d(cr_4;D) &\le d(cr_4;D) + \sum_{H \in \mc{F}_t} \alpha_H d(H;D) =
  \sum_{H \in \mc{F}_t} (d(cr_4;H) + \alpha_H) d(H;D) \\
  &\le \max_{H\in \mc{F}_t} \{d(cr_4;H) + \alpha_H\}.
\end{align*}

Since $\alpha_H$ can be negative for some models $H$, the hope is
that this will improve the low coefficients by transferring weight
from high coefficients. In order to find such inequalities, we need
another property of the flag densities.

\begin{proposition}
\label{prop:small_error}
If $\sigma$ is a fixed type of size $k$, $m\ge 1$ is
an integer, $\{F_i\}_{i=1}^m$ is a fixed family of $\sigma$-flags of sizes
$|F_i|=l_i$, and $l \ge k + \sum_{i=1}^m (l_i - k)$ is an integer, then for
any flag $F$ of order $n \ge l$, we have
\[
  p_{\sigma}(F_1,\ldots,F_m;F) = \left[\prod_{i=1}^m p_{\sigma}(F_i;F)\right] +
  O(1/n),
\]
where the constant in the big-O notation might depend on the family
$\{F_i\}_{i=1}^m$.
\end{proposition}

One can prove \propref{prop:small_error} by noting that, if we drop the
requirement that the sets $X_i$ are disjoint in the definition of
$p_{\sigma}(F_1,\ldots,F_m;F)$, the events $E_i$ will become
independent, and thus $\Prob\lt[\cap_{i=1}^m E_i\rt] = \prod_{i=1}^m
\Prob[E_i] = \prod_{i=1}^m p_{\sigma}(F_i;F)$.  The error
introduced is the probability that these sets $X_i$ will intersect
in $F$, which is $O(1 / n)$.  It is tempting to claim a similar product
formula for the unlabeled flag densities $d_{\sigma}$, but we cannot do so.
In the above equation, it is essential that all the $\sigma$-flags $F_i$
share the same labeled type $\sigma$, and hence we require $F$ to be a
$\sigma$-flag.

\medskip

We are now ready to establish some inequalities. Let's first fix a
type $\sigma$ of size $k$. If $Q$ is any positive semi-definite
$|\mc{F}_l^{\sigma}|\times|\mc{F}_l^{\sigma}|$ matrix with
rows and columns indexed by the same set $\mc{F}_l^{\sigma}$,
where $l\ge k$, define the ``quadratic form'' on flags by
\[
  Q\{\mc{F}_l^{\sigma}\} :=
  \sum_{F_1,F_2 \in \mc{F}_l^{\sigma}} Q_{F_1,F_2}
  F_1 \cdot F_2 \in \mc{A}^{\sigma}.
\]
\propref{prop:small_error} yields, for a $\sigma$-flag $F$
of sufficiently large size, the following approximation
\begin{equation}
  \label{eqn:approx_sdp}
  p_{\sigma}(Q\{\mc{F}_l^{\sigma}\}; F)\approx\sum_{F_1,F_2 \in
  \mc{F}_l^{\sigma}} Q_{F_1,F_2} p_{\sigma}(F_1;F)
  p_{\sigma}(F_2;F).
\end{equation}
Note that because $Q$ is positive semi-definite, the right hand side
of \eqref{eqn:approx_sdp} is always non-negative. Even after averaging
we obtain:
\begin{align*}
  [[Q]]_{\sigma}(D) &:=
  p([[Q\{\mc{F}_l^{\sigma}\}]]_{\sigma};D) =
  \sum_{F_1,F_2 \in \mc{F}_l^{\sigma}} Q_{F_1,F_2}
  d_{\sigma}(F_1,F_2;D) \\
  &= \sum_{F_1,F_2 \in \mc{F}_l^{\sigma}} Q_{F_1,F_2} \left(
  \sum_{F\in \mc{F}_{t}^{\sigma}} p_{\sigma}(F_1,F_2;F)
  d_{\sigma}(F;D) \right) \\
  &= \sum_{F\in \mc{F}_{t}^{\sigma}}
  \left(\sum_{F_1,F_2 \in \mc{F}_l^{\sigma}} Q_{F_1,F_2}
  p_{\sigma}(F_1,F_2;F)\right) d_{\sigma}(F;D) \\
  &= \sum_{F\in \mc{F}_{t}^{\sigma}} \left(\sum_{F_1,F_2 \in
  \mc{F}_l^{\sigma}} Q_{F_1,F_2} p_{\sigma}(F_1;F)
  p_{\sigma}(F_2;F)\right) d_{\sigma}(F;D) + O(1/n) \ge
  o_{n\to\infty}(1),
\end{align*}
where $n$ is the size of the tree-pair $D$ and $2l - |\sigma| \le t \le n$
is some fixed integer. Therefore, when $n$ is large, we have that
$[[Q]]_{\sigma}(D)$ is asymptotically non-negative. For each
admissible model $H$ of size exactly $t$, let $\alpha_H =
[[Q]]_{\sigma}(H)= \sum_{F_1,F_2 \in \mc{F}_t^{\sigma}}
Q_{F_1,F_2} d_{\sigma}(F_1,F_2;H)$. We then have
\[
 [[Q]]_{\sigma}(D) = \sum_{H\in \mc{F}_t} \alpha_H d(H;D)
 \ge o_{n\to\infty}(1).
\]
The expression in the middle of the above equation is called the
\emph{expansion} of $[[Q]]_{\sigma}(D)$ in tree-pairs of size $t$, with
$\alpha_H$ the coefficients of the expansion. For the sake of
conciseness, we often omit the parameter $D$ and express this
asymptotic inequality (combined with the expansion in size $t$) in
the syntax of flag algebras
\begin{equation}
  \label{eqn:positivity_flags}
  [[Q]]_{\sigma} := [[Q\{\mc{F}_l^{\sigma}\}]]_{\sigma}=
  \bigg[\bigg[\sum_{F_1,F_2\in\mc{F}^{\sigma}_l} Q_{F_1,F_2}
  F_1 \cdot F_2\bigg]\bigg]_{\sigma}=\sum_{H\in \mc{F}_t} \alpha_H H \ge 0.
\end{equation}
(Note that all inequalities between flags stated in the language of
flag algebras are asymptotic.)


\medskip

In general, if we have more than one inequality available, we can
combine them together, provided they are all expanded in the same
size $t$. Suppose we have $r$ inequalities given by the positive
semi-definite matrices $Q_i$ of the $\sigma_i$-flags of size $l_i$.
Adding them together, we obtain
\[
\sum_{i=1}^r [[Q_i]]_{\sigma_i} = \sum_{H \in \mc{F}_t} \alpha_H H \ge  0,
\]
where
\[
  \alpha_H = \sum_{i=1}^r \left(\sum_{F_1,F_2 \in \mc{F}_{l_i}^{\sigma_i}}
  (Q_i)_{F_1,F_2} d_{\sigma_i}(F_1,F_2;H)\right),
\]
and we want to minimize $\max_{H \in \mc{F}_t} \left\{d(cr_4;H) +
\alpha_H \right\}$.

\medskip

Thus we have transformed the original problem of finding a minimum
upper bound for $d(cr_4;G)$ into a linear system involving the
variables $(Q_i)_{F_k,F_l}$.  As we have the constraint that the
matrices $Q_i$ should be positive semi-definite, this is a
semi-definite programming problem.  To take the maximum coefficient
in the expansion, we introduce an artificial variable $y$, and require
it to be bounded below by all the coefficients. Hence we have the
following SDP problem in the variables $y$ and $( Q_i )_{F_1,F_2}$:

\medskip

Minimize $y$, subject to the constraints:
\begin{itemize}
\item We have $s_H \ge 0$ for all $H\in \mc{F}_t$, where
\begin{equation}
\label{eqn:surplus}
s_H := y-d(cr_4;H) - \sum_{i=1}^r \lt(\sum_{F_1,F_2 \in
\mc{F}_{l_i}^{\sigma_i}} (Q_i)_{F_1,F_2} d_{\sigma_i}(F_1,F_2;H)\rt).
\end{equation}
The variables $s_H$ are called \emph{surplus} variables.
\item $Q_i$ is positive semi-definite for $i \in [r]$. (The matrices $Q_i$
are often called the \emph{block variables} of the SDP problem.
We can assume without loss of generality that each $Q_i$ is symmetric,
as otherwise we could replace $Q_i$ by $(Q_i + Q_i^T)/2$.)
\end{itemize}

\medskip

A computer can solve this SDP problem numerically, allowing for an
efficient determination of the inequalities required to prove the
extremal problem. We note at this point, that the solution to the SDP
problem need not only give the asymptotic bound, but can also provide
some structural information about the extremal tree-pair.

\section{Main result}
\label{sec:main}

In this section we discuss some practical considerations of the main
theorem. For a square matrix $A$, let $\tr(A)$ denote its trace.
The original formulation of the SDP problem can be rewritten in
concise matrix notation as follows:
\begin{equation}
  \label{eqn:sdp_prob}
  \begin{array}{rrcll}
  \text{minimize} & \tr (C \cdot Q)  \\
  \text{subject to} &  \tr(A_j \cdot Q) &=& b_j, &
  \text{for all } j=1,\ldots, m, \\
  \text{and} & Q  &\succeq& 0& \text{(i.e., $Q$ is positive semi-definite)}
  \end{array}
\end{equation}
where $m=|\mc F_t|$ represents the number of constraints in the problem,
$C$ is the cost matrix (we have $\textup{tr}(C\cdot X) = y$, where
$y$ is as in the previous subsection), $Q$ is positive semi-definite
matrix consisting of all the block-variable matrices $Q_i$, and
each equation $\textup{tr}(A_j \cdot Q) = b_j$ corresponds to one
of the equations \eqref{eqn:surplus} from the original formulation.
In particular, if $\mc F_t =\{H_1,\ldots, H_m\}$, we have $b_j =
d(cr_4;H_j)$. Finally, we let $\ell$ denote the number of rows/columns
of $Q$.

A computer usually cannot solve \eqref{eqn:sdp_prob} exactly, but
only approximately. In other words, the output of the SDP solver will be
a matrix $Q'$ that satisfies the constraints approximately. Namely, we have
\begin{equation}
\label{eqn:sdp_approx}
\begin{array}{rl}
\big|\tr (A_j\cdot Q') - b_j\big| &\le \eps, \quad\text{for }j=1,\ldots m,
  \text{ and} \\
Q' + \eps I_\ell &\succeq 0,
\end{array}
\end{equation}
for some small $\eps > 0$ (usually $\eps < 10^{-9}$), where $I_\ell$
denotes the $\ell\times \ell$ identity matrix. In what follows we describe
how to obtain a matrix $Q$ that satisfies all the constraints of
\eqref{eqn:sdp_prob} and is not ``too far'' from the approximate
solution $Q'$. That way $\tr(C\cdot Q) \approx \tr(C\cdot Q')$.

A natural first step towards this goal is to slightly change $Q'$ so that it
satisfies all the linear constraints in \eqref{eqn:sdp_prob}. For that
purpose, we will project $Q'$ to the affine subspace of all
$\ell \times \ell$ symmetric real matrices $Q$ that satisfy
$\tr(A_j\cdot Q)=b_j$ for all $j=1,\ldots, m$.
Let $Q''$ denote this projection. How much did we change the approximate
solution? Namely, how large is $||Q'- Q''||_\infty$? We recall that
for a matrix $A$, we denote $||A||_\infty := \max_{ij} |A_{ij}|$ and
$||A||_1 := \sum_{ij} |A_{ij}|$.

To estimate $||Q'- Q''||_\infty$ we often use some inequalities
from the following proposition.

\begin{proposition}
\label{prop:inequalities}
The following statements are true:
\begin{enumerate}[(i)]
\item If $A\in \mathbb{R}^{m\times n}$ and $B\in \mathbb{R}^{n \times l}$
are two real matrices, then $||A\cdot B||_{\infty} \le n \cdot ||A||_\infty
\cdot ||B||_\infty$.
\item If $A\in \mathbb{R}^{m\times n}$ and $v\in \mathbb{R}^n$, then
$||A\cdot v||_\infty \le ||A||_\infty \cdot ||v||_1$.
\item Let $A\in \mathbb{R}^{n\times n}$ be any $n\times n$ symmetric matrix.
Then $A + n\cdot ||A||_\infty \cdot I_n$ is positive semi-definite.
\end{enumerate}
\end{proposition}
\begin{proof}
\begin{enumerate}[(i)]
\item Let $C = A\cdot B$. We have $C_{ij} = \sum_{k=1}^n
A_{ik} B_{kj}$, thus
\[
 |C_{ij}| \le \sum_{k=1}^n
  |A_{ik}| |B_{kj}| \le \sum_{k=1}^n ||A||_\infty ||B||_\infty =
  n \cdot ||A||_\infty \cdot ||B||_\infty,
\]
hence $||C||_\infty \le n \cdot ||A||_\infty \cdot ||B||_\infty$.
\item Let $w = A\cdot v$. We have $w_i = \sum_{j=1}^n A_{ij}v_j$, thus
\[
 |w_i| \le \sum_{j=1}^n
  |A_{ij}| |v_j| \le \sum_{j=1}^n ||A||_\infty |v_j|=
  ||A||_\infty \cdot ||v||_1,
\]
therefore $||w||_\infty \le ||A||_\infty \cdot ||v||_1$.
\item Let $B=A + n\cdot ||A||_\infty \cdot I_n$. It suffices to show that
for all $v\in \mathbb{R}^n$, we have $v^T\cdot B \cdot v \ge 0$. Let
$a = v^T \cdot A \cdot v$. Using the definition of $B$, we obtain
\[
  b:=v^T\cdot B \cdot v = v^T \cdot A \cdot v + n\cdot ||A||_\infty\cdot ||v||_2^2
  = a + n\cdot ||A||_\infty\cdot ||v||_2^2.
\]
By (ii) applied twice, we infer
\[
  |a| = ||v^T \cdot A \cdot v||_\infty \le ||v^T||_1 \cdot||A\cdot v||_\infty
   \le ||A||_\infty\cdot ||v||_1^2
\]
So by Cauchy-Schwarz inequality, we obtain $|a| \le ||A||_\infty\cdot ||v||_1^2
\le n\cdot ||A||_\infty\cdot ||v||_2^2$, therefore $b \ge 0$, finishing the proof.
\end{enumerate}
\end{proof}

In what follows, we introduce further notation in order to express $Q''$
in terms of $Q'$ and the parameters of the problem \eqref{eqn:sdp_prob}.
Let $\mc S$ be the linear space of all $\ell\times \ell$
real symmetric matrices, and let $\mathbf{A}$ be the linear map
$\mathbf{A}: \mc S \to \mathbb{R}^m$ defined by $\mathbf{A}(Q)_j =
\tr(A_j \cdot Q)$. In addition, let $b\in \mathbb{R}^m$ be vector
with coordinates $b_j$ for $j=1,\ldots, m$, and let $\mc H$ be the
affine subspace of all $\ell\times \ell$ real symmetric matrices $Q$
that satisfy the linear constraints of \eqref{eqn:sdp_prob}, namely
$\tr(A_j\cdot Q)=b_j$ for $j=1,\ldots, m$.
Note that $\mc H$ is the pre-image of $b$ by $\mathbf{A}$. Let $\mathbf{P}$
be the orthogonal projection from the set $\mc S$ to the affine subspace
$\mc H$. One can compute this projection by a solution to a least squares
problem as follows:
\[
  \mathbf{P}(Q) = Q + \mathbf{A}^T \cdot (\mathbf{A} \cdot \mathbf{A}^T)^{-1}
   (b - \mathbf{A}(Q)),
\]
for all $Q\in \mc S$, where $\mathbf{A}^T: \mathbb{R}^m \to \mc S$ denotes
the transpose of $\mathbf{A}$. We have $Q'' = \mathbf{P}(Q')$, thus by
\propref{prop:inequalities}~(i), we have
\[
  ||Q''-Q'||_\infty \le m \cdot ||\mathbf{A}^T \cdot (\mathbf{A} \cdot
  \mathbf{A}^T)^{-1}||_\infty \cdot ||b - \mathbf{A}(Q')||_\infty
  \le \eps m \cdot ||\mathbf{A}^T \cdot (\mathbf{A} \cdot
  \mathbf{A}^T)^{-1}||_\infty.
\]
For all the instances of \eqref{eqn:sdp_prob} that we consider, one can verify
that $||\mathbf{A}^T \cdot (\mathbf{A} \cdot \mathbf{A}^T)^{-1}||_\infty \le 1$,
thus $\eps':=||Q''-Q'||_\infty < \eps m$. This inequality together with
\propref{prop:inequalities}~(ii) implies that
\[
  \tr(C\cdot Q'') = \tr(C\cdot Q') + \tr(C\cdot (Q''-Q'))\le
  \tr(C\cdot Q') + \eps m \cdot ||C||_1.
\]

We know that $Q''$ satisfies all the linear constraints of \eqref{eqn:sdp_prob},
but $Q''$ might not be positive semi-definite. An application of
\propref{prop:inequalities}~(iii) yields that
\[
  (Q'' - Q') + \ell\cdot ||Q''-Q'||_\infty \cdot I_\ell \succeq 0,
\]
which, together with the inequality $Q'+\eps I_\ell\succeq 0$ from
\eqref{eqn:sdp_approx}, implies that $Q'' + (\eps + \ell \eps')\cdot I_\ell
\succeq 0$. To make $Q''$ positive semi-definite, we hope to find
a matrix $\tilde Q$ that satisfies $\tr(A_j \cdot \tilde Q) = 0$
for $j=1,\ldots, m$ and such that all the eigenvalues of $\tilde Q$
are large. If such $\tilde Q$ exists then $Q'' + \delta \tilde Q$ will
be positive semi-definite for some small $\delta >0$ and will also
satisfy the linear constraints in \eqref{eqn:sdp_prob}. For this reason,
we consider the following problem:

\begin{equation}
  \label{eqn:sdp_prob2}
  \begin{array}{rrcll}
  \text{minimize} & \tr (0 \cdot \tilde Q)  \\
  \text{subject to} &  \tr(A_j \cdot \tilde Q) &=& 0, &
  \text{for all } j=1,\ldots, m, \\
  \text{and} & \tilde Q  &\succ& 0 &\text{(i.e., $\tilde Q$
  is strictly positive-definite)} \\
  \end{array}
\end{equation}

Note that the function being minimized is the constant zero function,
so \eqref{eqn:sdp_prob2} is a pure feasibility problem.
We again use computers to obtain an approximate solution $\tilde Q'$
to \eqref{eqn:sdp_prob2}. Surprisingly, it turns out that the obtained
solution $\tilde Q'$ not only satisfies $\big|\tr(A_j \cdot \tilde Q')|
< \eps$ for all $j=1,\ldots, m$, but also has a large smallest eigenvalue
(much larger than $\eps + \ell \eps'$), even though
$|\tr(C\cdot \tilde Q')|$ is relatively small. We will later exploit
these properties to adjust $Q''$ to an exact solution of
\eqref{eqn:sdp_prob}.

Using similar ideas as before, we obtain a matrix $\tilde Q''$ that satisfies
$\tr(A_j \cdot \tilde Q'')=0$ for all $j=1,\ldots, m$ by means of
orthogonal projection of $\tilde Q'$ to the appropriate subspace.
As we have already seen, this operation only slightly changes the eigenvalues
of $\tilde Q'$. Finally we let $Q:= Q'' + \delta \tilde Q''$,
where $\delta= \frac{\eps + \ell \eps'}{\lambda}$ and $\lambda$ is
the smallest eigenvalue of $\tilde Q''$ (in all of our instances we have
$\delta < 10^{-4}$). Clearly $Q$ satisfy all the constraints
of \eqref{eqn:sdp_prob}, including $Q \succeq 0$. Moreover, we have
\[
  \tr(C\cdot Q) = \tr(C\cdot Q'') + \delta\cdot \tr(C\cdot \tilde Q'')
   \le \tr(C\cdot Q')+\eps m \cdot ||C||_1+
   \delta\cdot \tr(C\cdot \tilde Q''),
\]
and since both $\eps m$ and $\delta$ are typically small, we will not
change the objective value much from the original approximate solution
$Q'$ to the exact solution $Q$. Therefore $Q$ is the desired exact
solution which is ``close'' to $Q'$.

In what follows we have a compiled table displaying the several
bounds obtained for different instances of the SDP problem. 
The first column represents the parameter $t$, which is the size
of the tree-pairs used in the expansion of the problem
(see \secref{sec:extremal_prob} for more details). The second
column counts the number of tree-pairs of size $t$. The third
column indicates how many types where used, that is, the types
$\sigma$ for the inequalities of the form \eqref{eqn:positivity_flags}.
The used types are all those having size with the same parity as
and strictly smaller than $t$. The fourth column contains the total
number of variables in the SDP instance including the surpluses.
Finally, the last column tells the bound obtained from the SDP solver.  
The program used to generate the SDP instances and verify these
calculations can be downloaded at
\url{http://http://http://arxiv.org/src/1505.04344v2/anc}.

\begin{table}[h]
\centering
\begin{tabular}{|l||l|l|l|l|}
\hline
$t$ & $m = |\mc F_t|$
& \# of types used & $\ell=$\# of variables& Bound \\
\hline
5 &  4 & 1 & 50 & 0.884766\\
6 &  31 & 3 & 697 & 0.760257\\
7 &  243 & 6 & 12050 & 0.707633\\
8 &  3532 & 35 & 506171 & 0.688397\\
\hline
\end{tabular}
\caption{Several instances of the main SDP problem.}
\label{tbl:main}
\end{table}
\textbf{Remark.} It was still possible to run the program
for $t=9$, but unfortunately the bound was still strictly
greater than $2/3$.

\section{On caterpillar trees}
\label{sec:caterpillar}

In this section, we prove \thmref{thm:caterpillar} --- a restricted
version of \conjref{conj:main} to caterpillar trees. One possible
approach to this problem is to use the same machinery of flag algebras
for the theory of tree-pairs restricted to caterpillar trees, and try
to obtain a bound in the same way as we did for \thmref{thm:main}.
However, this approach does not immediately yield the bound of $\frac 23$,
and so it is necessary (and worthwhile) to think about this problem from
a different perspective. In the next few paragraphs we will explain
how to map the problem of computing the induced density of $cr_4$
in a tree-pair of caterpillar trees into a problem of counting induced
sub-permutations of size $4$.

Suppose $D=\{\ovl{T}_1, \ovl{T}_2\}$ is a tree-pair composed by two caterpillar
trees on $n+2$ leaves (as exemplified in \figref{fig:caterpillar}).
One can think of $D$ as a permutation of $\{\alpha,x_1,\ldots,
x_n,\beta\}$ --- a permutation that tells us exactly how the leaves of
$T_1$ are ``attached'' to the leaves of $T_2$.
For instance, the tree-pair $cr_5^A$ in \figref{fig:F_5} could be
represented by the permutation $\alpha \to \alpha, x_1 \to x_1, x_2
\to x_3, x_3\to x_2, \beta \to \beta$.
In fact, multiple permutations might give rise to the same tree-pair.
Regarding this matter, our first observation is that any caterpillar
tree on four or more leaves has exactly $8$ distinct
automorphisms. To illustrate this observation, consider the caterpillar
tree on $n+2$ leaves labelled by $\alpha,x_1,\ldots, x_n,\beta$ as depicted
in \figref{fig:auto_cater}. One of the automorphisms of this tree is
$\sigma_1$, which is the unique automorphism that maps $\alpha$ to $\beta$
and $\beta$ back to $\alpha$, such as in a ``reflection''. Similarly, $\sigma_2$
is the automorphism that only swaps $\alpha$ with $x_1$ and leaves all the
remaining vertices in place. Finally, $\sigma_3$ is the automorphism that swaps
$\beta$ and $x_n$. The group of automorphisms can be then written as
$\{\sigma_1^{i_1} \sigma_2^{i_2} \sigma_3^{i_3} : 0 \le i_1, i_2, i_3 \le 1\}$.
Our second observation is that given a permutation $\pi$, and two
automorphisms $\sigma, \sigma'$ of the caterpillar tree with leaves
labelled $\{\alpha, x_1,\ldots, x_n, \beta\}$, the permutation
$\sigma \pi \sigma'$ represents the same tree-pair as $\pi$ itself.
Here we think of $\sigma$ and $\sigma'$ as only acting solely on
the leaves of the caterpillar trees.

\begin{figure}[H]
\centering
\tikzset{
  pil/.style={
    ->,
    thick,
    shorten <=2pt,
    shorten >=2pt,}
}
\begin{tikzpicture}
  [scale=1,auto=left,every node/.style={circle,draw,fill=white,inner sep=1pt}]
  \node (nleft) at (-4,  0 cm) {$\alpha$};
  \node (n1) at (-3, -1 cm) {$x_1$};
  \node[inner sep=0pt] (m1) at (-3,  0 cm) {};
  \node (n2) at (-2, -1 cm) {$x_2$};
  \node[inner sep=0pt] (m2) at (-2, 0 cm) {};
  \node (n3) at ( 3, -1 cm) {$x_n$};
  \node[inner sep=0pt] (m3) at ( 3,  0 cm) {};
  \node[draw=none,fill=none] (n4) at (0.5, -1 cm) {$\ldots$};
  \node (nright) at ( 4,  0 cm) {$\beta$};

  \draw (nleft) -- (m1);
  \draw (m1) -- (m2);
  \draw (m2) -- (m3);
  \draw (m3) -- (nright);
  \draw (m1) -- (n1);
  \draw (m2) -- (n2);
  \draw (m3) -- (n3);

  \path (nleft) edge [pil, <->, bend right]
    node[below, draw=none, fill=none] {$\sigma_2$} (n1) ;

  \path (nright) edge [pil, <->, bend left]
    node[below, draw=none, fill=none] {$\sigma_3$} (n3) ;

  \path (-3.5, .7cm) edge [pil, <->]
    node[above, draw=none, fill=none] {$\sigma_1$} (3.5, .7cm) ;

\end{tikzpicture}
\caption{The automorphisms of a caterpillar tree.}
\label{fig:auto_cater}
\end{figure}
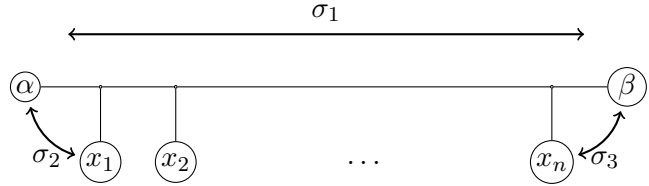

Given a permutation $\pi$ of $L:=\{\alpha, x_1,\ldots,x_n, \beta\}$, how
do we count the number of induced copies of $cr_4$ in the corresponding
tree-pair $D$ represented by $\pi$? Suppose $S\subseteq L$ is a subset
of size $4$ such that $\alpha,\beta \not\in S$ and $\alpha,
\beta\not\in\pi(S)$, say $S=\{x_{i_1}, x_{i_2}, x_{i_3}, x_{i_4}\}$
with $\pi(x_{i_t}) = x_{j_t}$ for $t=1,\ldots, 4$ and $i_1<i_2<i_3<i_4$.
Here it is helpful to think of $S$ as a subset of the leaves of
$T_1$ before the identification with the leaves of $T_2$.
The corresponding leaves selected by $S$ will induce a copy
of $id_4$ in $D$ if either $\max\{j_1, j_2\} < \min\{j_3, j_4\}$, or
$\max\{j_3, j_4\} < \min\{j_1, j_2\}$. Otherwise $S$ induces a copy
of $cr_4$. Since there are only $O(n^3)$ subsets $S$ that
do not satisfy the condition $\{\alpha,\beta\}\cap (S\cup \pi(S))
=\emptyset$, the problem of computing the density of $id_4$ in $D$
essentially becomes the problem of computing the density of the following
induced sub-permutations in a permutation $\pi \in S_n$:
\begin{equation}
  \label{eqn:4densities}
  1234, 1243, 2134, 2143, 3412, 4312, 3421, 4321.
\end{equation}

The machinery of flag algebras is very general and thus 
can also be applied
to the theory of permutations. In fact, we have the following theorem
which implies \thmref{thm:caterpillar}.

\begin{theorem}

  \label{thm:permutations}
  The sum of the densities of the permutations listed in
  \eqref{eqn:4densities} inside a permutation $\pi\in S_n$
  is at least $\frac 13 + o(1)$ for $n$ large.

\end{theorem}

\begin{proof}
Using the notation from flag algebras, let $\phi$ denote the sum of
the densities of the permutations in \eqref{eqn:4densities}, that is,
\[
  \phi = 1234+1243+2134+2143+3412+4312+3421+4321.
\]
In this notation, a flag is just a permutation with some entries labeled
by the set $[k]$ for some $k\ge 0$.
For instance $1\circled 6_2 42\circled 3_15$ denotes a flag whose underlying
permutation is $164235$ for which the fifth entry is labeled $1$ and
the second entry is labeled $2$. In this case, the type of the flag is
$\circled 2_2 \circled 1_1$, since the sub-permutation induced by the labeled
entries is isomorphic to $21$ (corresponding to the entries $63$).
Another example is the flag $\circled 5_3 17\circled 2_13\circled 6_24$ ---
its underlying permutation is $5172364$ and its type is $\circled 2_3
\circled 1_1\circled 3_2$. With this definitions in mind, we remark that
a type is just a permutation of $[k]$ with all entries labeled by elements
of the set $[k]$. Thus, for an integer $k\ge 0$, types on $k$ entries are
in one-to-one correspondence with pairs of permutations of $[k]$.

Consider the following $4$ types of size $2$
\[
  \begin{array}{lclclcl}
    \rho_1 &=& \circled 1_1\circled 2_2, &\quad&
    \rho_2 &=& \circled 2_1\circled 1_2,\\
    \rho_3 &=& \circled 1_2\circled 2_1, &\quad&
    \rho_4 &=& \circled 2_2\circled 1_1.\\
  \end{array}
\]

We have
\begin{equation}
  \label{eqn:proof_perm}
  \phi = \frac 13 + \sum_{i=1}^4 3\cdot[[(X_i-Y_i)^2]]_{\rho_i}
  + 6\cdot [[(X_i-Z_i)^2]]_{\rho_i}
\end{equation}
where
\begin{align*}
 X_1 &= -1\circled 2_1\circled 3_24 -1\circled 2_1\circled 4_23
     +4\circled 2_1\circled 3_21+3\circled 2_1\circled 4_21, \\
 Y_1 &= +1\circled 2_13\circled 4_2 + 1\circled 2_14\circled 3_2
    - 4\circled 2_11\circled 3_2 - 3\circled 2_11\circled 4_2, \\
 Z_1 &= +\circled 2_11\circled 3_24 + \circled 2_11\circled 4_23
    - \circled 2_1 4\circled 3_21 - \circled 2_1 3 \circled 4_2 1,\\
 X_2 &= -4\circled 3_1\circled 2_21 -4\circled 3_1\circled 1_22
     +1\circled 3_1\circled 2_24+2\circled 3_1\circled 1_24, \\
 Y_2 &= +4\circled 3_12\circled 1_2 + 4\circled 3_11\circled 2_2
    - 1\circled 3_14\circled 2_2 - 2\circled 3_14\circled 1_2, \\
 Z_2 &= +\circled 3_14\circled 2_21 + \circled 3_14\circled 1_22
    - \circled 3_1 1\circled 2_24 - \circled 3_1 2 \circled 1_2 4,\\
 X_3 &= -1\circled 2_2\circled 3_14 -2\circled 1_2 \circled 3_14
     +4\circled 2_2\circled 3_11 + 4\circled 1_2\circled 3_12, \\
 Y_3 &= +\circled 1_22\circled 3_14 + \circled 2_2 1\circled 3_14
    - \circled 2_2 4\circled 3_11 - \circled 1_2 4\circled 3_12, \\
 Z_3 &= +1\circled 2_2 4\circled 3_1 + 2\circled 1_2 4\circled 3_1
    - 4\circled 2_2 1\circled 3_1 -  4\circled 1_2 2\circled 3_1,\\
 X_4 &= -4\circled 3_2\circled 2_11 -3\circled 4_2 \circled 2_11
     +1\circled 3_2\circled 2_14 + 1\circled 4_2\circled 2_13, \\
 Y_4 &= +\circled 4_23\circled 2_11 + \circled 3_2 4\circled 2_11
    - \circled 3_2 1\circled 2_14 - \circled 4_2 1\circled 2_13, \\
 Z_4 &= +4\circled 3_2 1\circled 2_1 + 3\circled 4_2 1\circled 2_1
    - 1\circled 3_2 4\circled 2_1 -  1\circled 4_2 3\circled 2_1,
\end{align*}
therefore $\phi \ge \frac 13$, thereby proving the theorem.
Note that in order to attest the correctness of \eqref{eqn:proof_perm},
it suffices to evaluate the left- and the right-hand side of the equation
for all permutations of size 6. 
\end{proof}

\section{Concluding remarks}
\label{sec:conclusion}

In \thmref{thm:main} we showed that the maximum quartet distance
between two arbitrary phylogenetic trees on $n$ leaves is at most
$(0.69+o(1)) \binom{n}{4}$. It would be interesting to know if 
the techniques of this paper can be pushed even further to
obtain the $(\frac 23+o(1))\binom{n}{4}$ thereby establishing
\conjref{conj:main}. 

Another approach to \conjref{conj:main} is to solve an extremal
problem in the theory of $4$-uniform hypergraphs. In \cite{AlonSnirYuster},
Alon \emph{et al} proved the asymptotic upper bound of
$\frac{9}{10}\binom{n}{4}$
by mapping a tree-pair into a $4$-uniform hypergraph in the following way.
The vertices of the hypergraph are the leaves of the tree-pair and a
subset $S$ of $4$ leaves is an edge of the hypergraph if the sub-tree-pair
induced by $S$ is isomorphic to $cr_4$. They showed that the resulting
hypergraph $\mc H$ does not contain a copy of $K_6^4$ --- the complete
$4$-uniform hypergraph on $6$ vertices. One remark is that not only $K_6^4$
but also several other forbidden hypergraphs do not appear as induced
subgraphs of $\mc H$. A natural question emerges: can one characterize
this family of forbidden subgraphs? 
In particular, is it finite? 
\vspace{0.1cm}

\noindent
{\bf Acknowledgement}
We thank Sagi Snir and Raphy Yuster for helpful discussions.

\end{document}